\RequirePackage{fix-cm}
\documentclass[twocolumn]{svjour3}          
\smartqed  
\usepackage{graphicx}
\usepackage{booktabs}
\usepackage{cite}
\usepackage{times}
\usepackage{epsfig}
\usepackage{graphicx}
\usepackage{amsmath}
\usepackage{amssymb}
\usepackage{caption}
\usepackage{adjustbox}
\usepackage{algorithm}
\usepackage{algpseudocode}
\usepackage[round]{natbib}

\begin{document}

\title{Eyes on the Prize: Improving Biological Surface Registration via Forward Propagation
}


\author{Robert J. Ravier
}


\institute{Robert J. Ravier \at
              Department of Electrical and Computer Engineering \\
              Duke University \\
              \email{robert.ravier@duke.edu}           
}

\date{Received: date / Accepted: date}

\maketitle

\begin{abstract}
Many algorithms for surface registration risk producing significant errors if surfaces are significantly nonisometric. Manifold learning has been shown to be effective at improving registration quality, using information from an entire collection of surfaces to correct issues present in pairwise registrations. These methods, however, are not robust to changes in the collection of surfaces, or do not produce accurate registrations at a resolution high enough for subsequent downstream analysis. We propose a novel algorithm for efficiently registering such collections given initial correspondences with varying degrees of accuracy. By combining the initial information with recent developments in manifold learning, we employ a simple metric condition to construct a measure on the space of correspondences between any pair of shapes in our collection, which we then use to distill soft correspondences. We demonstrate that this measure can improve correspondence accuracy between feature points compared to currently employed, less robust methods on a diverse dataset of surfaces from evolutionary biology. We then show how our methods can be used, in combination with recent sampling and interpolation methods, to compute accurate and consistent homeomorphisms between surfaces.
\end{abstract}

\section{Introduction}\label{sec:intro}
Shape registration is a fundamental problem in computer vision. Accurate registrations are required for many applications, such as object retrieval, statistical shape analysis, medical imaging, among others \cite{guo20143d, dryden2016statistical, styner2006framework, tam2012registration}. A plethora of algorithms for registering pairs of shapes have been proposed, many of which can be directly formulated as minimizing some well-defined objective function (see \cite{joshi2000landmark, myronenko2010point, van2011survey, ovsjanikov2012functional, lipman2013continuous}) for a small sample of such works. For practical applications involving collections of shapes, it is often necessary to have a {\emph{consistent}} set of registrations: given shapes $\{S_{i}\}_{i \in \mathcal{I}},$ the collection of registrations $\{\phi_{ji}\}_{i,j \in \mathcal{I}}$ is consistent if $\phi_{ik} \circ \phi_{kj} \circ \phi_{ji} = \text{Id}_{S_{i}}$ for all $i,j,k \in \mathcal{I},$ where $\text{Id}_{S_{i}}$ is the identity map on $S_{i}.$ Consistency is a requirement for posthoc statistical analysis \cite{lorenzi2013geodesics}, and is generally obtained by a refinement of an initially computed collection of registrations, either through optimization \cite{huang2013consistent} or by directly defining a consistent collection based on composition \cite{aigerman2016hyperbolic, boyer2015new, puente2013distances}.

\begin{figure}[h]
\begin{center}
\includegraphics[scale=.38]{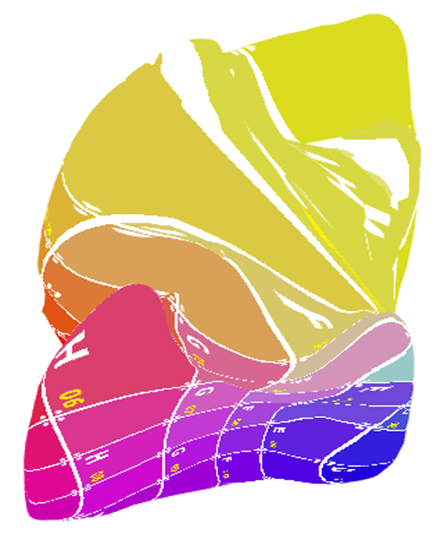}
\includegraphics[scale=.38]{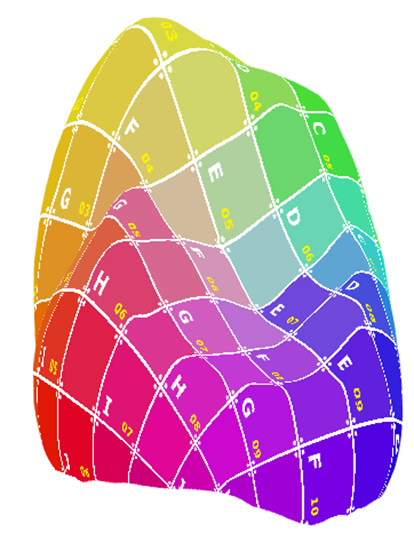}
\end{center}
\caption{A poor quality texture map between two primate teeth from the dataset in \cite{boyer2011algorithms} computed from the method in \cite{lipman2013continuous}.}
 \label{fig:badTextures}
 \vspace{-0.2in}
\end{figure}

Anatomical surfaces studied in evolutionary biology, the shapes that serve as the main motivation for this work, have historically proven to be difficult for algorithms to accurately register, especially if one is interested in homeomorphism. In contrast to those considered in medical biology, surfaces considered in evolutionary biology often consist of multiple species \cite{boyer2011algorithms}, which in turn often leads to high amounts of nonisometry. Though registration of nonisometric shapes is of increasing research interest \cite{pons2015dyna, rodola2017partial, schonsheck2018nonisometric, dyke2019shrec}, the mechanisms of nonisometry studied in recent works are generally distinct from those in evolutionary biology. In this setting, the main difficulty stems from both the lack of complexity of shape compared to those often considered, potential differences in the amount and location of regions with high curvature, as well as the potential type of curvature itself.  Figure~\ref{fig:badTextures} shows an example of a failure of the algorithm in \cite{lipman2013continuous} to register two primate molars from \cite{boyer2011algorithms}, for which ground truth registrations are available. In this case, the nonisometry between the two surfaces, occurring in the yellow regions of both texture maps, has resulted in a failure to find a good registration, as represented by the extreme stretching on the left textured surface. 

The failure in Figure~\ref{fig:badTextures} is perhaps unsurprising to the reader: the geometries of the top portions of both surfaces are quite distinct in the sense that some of the features on the left surface do not have obvious counterparts on the right surface. Though this stark difference naturally begs the question as to whether a meaningful correspondence is feasible to compute without significant uncertainty, we again recall from the previous paragraph that {\emph{well-accepted ground truth correspondences exist.}} \cite{boyer2011algorithms}. Correspondences for this example and other datasets within evolutionary biology are determined by experts manually selecting such features (termed {\emph{landmarks}} \cite{roth1993three, bookstein1997morphometric}) based on both geometric and non geometric biological information from the collection as a whole. Corresponding landmarks on pairs of disparate shapes are often determined by composing correspondences along paths of intermediates, where the degree of nonisometry between any two intermediates is smaller than those of the initial shapes of interest. This procedure matches recent findings within the registration literature \cite{puente2013distances, boyer2015new, pons2015dyna, bogo2017dynamic, gao2018development}, which suggest that the accuracy of registrations (and, in the case of \cite{puente2013distances, boyer2015new}, rigid alignments) between disparate pairs of shapes can be improved by using compositions of registrations between similar shapes. Here, the methods in \cite{puente2013distances, boyer2015new, gao2018development} define similar by an explicit notion of distance, while the methods in \cite{pons2015dyna, bogo2017dynamic} use time difference as a similarity measure, appropriate as the data of interest consists of motion capture of humans. 

All methods mentioned in the previous paragraph yield consistent registrations amongst all shapes in a collection; this is due to explicit choices of particular sequences of intermediate for each pair of shapes that result in consistency. As noted in the beginning of the section, this is not strictly necessary: the method proposed in \cite{huang2013consistent} allows one to obtain consistent registrations from an initial set of pairwise ones by means of solving an optimization problem. We note that the methodology proposed in \cite{huang2013consistent} makes the intrinsic and completely reasonable assumption that the initial pairwise registrations are of sufficient quality. For disparate pairs of shapes in our setting, this requires a sequence of intermediates over which to compute and compose registrations per experimental evidence. Though one could use additional information not immediately present within the collection of shapes of interest, it is likely that such information already results in consistent registrations, defeating the point of the optimization method in this setting.

\begin{figure*}[h]
\begin{center}
\includegraphics[width=0.45\textwidth]{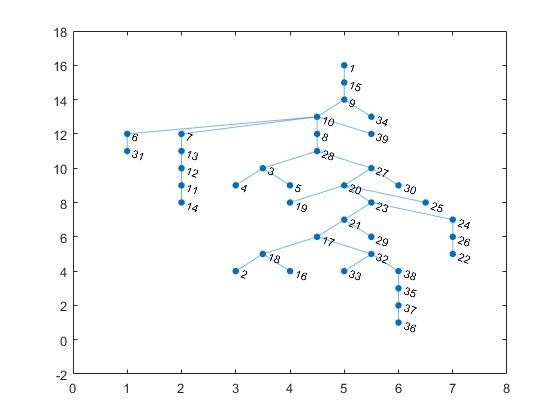}
\hfill
\includegraphics[width=0.45\textwidth]{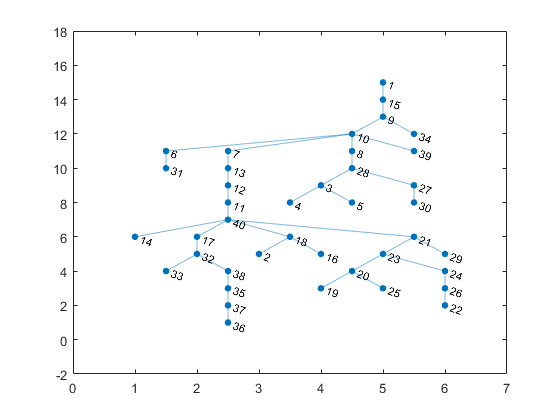}
\end{center}
\caption{Two minimum spanning trees used for consistent registration of shapes from \cite{boyer2011algorithms}. Left: the topology of the MST for 39 shapes with respect to the continuous Procrustes distance \cite{lipman2013continuous}. Right: the topology of the MST for a collection of 40 shapes consisting of the previous 39 with an additional one added.}
 \label{fig:MSTBad}
 \vspace{-0.2in}
\end{figure*}

The distance-based methods for achieving consistent registrations  \cite{puente2013distances, boyer2015new, gao2018development} may still not yield ones that are of sufficient quality between disparate shapes. Per the analysis of \cite{gao2018development}, accuracy of registration was observed to need not perfectly correlate with composition over small distances. Furthermore, there are few stability guarantees in this setting: for example, minimum spanning trees, which are used in \cite{puente2013distances, boyer2015new} as a means of determining consistent registrations, may wildly change topology if an additional sample is added or removed. An example of this can be found in Figure~\ref{fig:MSTBad}, meaning that incorporation of additional samples for analysis can cause monumental issues in registration. Addressing this issue, namely obtaining consistent, high-quality registrations in the nonisometric scenario, is the focus of our work.

An obvious way to solve this would be to improve the quality of all registrations, or at enough to generate a consistent collection via some heuristic. Based on the aforementioned experimental evidence, it is thus natural to define an optimization problem {\emph{for each pair of shapes}} in the collection, where the domain of the objective is over the space of all possible sequences of intermediates. This is naturally posed within the framework of combinatorial optimization, for which many solution techniques exist. The main problem is the choice of objective function, which is extremely nontrivial. Though one could minimize an energy functional that depends on a given pair of shapes and the registration between them, such as the continuous Procrustes distance \cite{lipman2013continuous}, as previously mentioned, smaller values of such energies do not perfectly correlate with increased accuracy of registration \cite{gao2018development}. At the same time, it is also infeasible to use direct accuracy measures as an objective function, as this implies it is known a priori, i.e. the problem is already solved. Thus, even in a quite general framework, we apparently stuck in a catch-22. The goal of this paper is to propose a mathematically principled solution to this problem by eschewing the language of optimization in favor of that of differential geometry and manifold learning.

In this paper, we propose a novel solution to the problem of registering disparate shapes in a collection. Rather than directly optimizing an objective function, we instead consider a surrogate problem of computing {\emph{parallel transports along geodesics of a manifold.}} Geodesic parallel transports have been previously considered in the realm of registration, most notably for medical shape analysis as a means of transforming data to a common template on which to perform statistical analysis \cite{lorenzi2013geodesics, pennec2019riemannian}. We will instead based our approach off recent advances within the realms of manifold learning \cite{gao2015hypoelliptic, gao2019diffusion}, in which registrations themselves are modeled as parallel transports on a {\emph{fiber bundle}}. We will use this model to significantly refine our search space to registrations induced by {\emph{eyes on the prize}} flows (EOP flows) within the shape space. This allows us to naturally arrive at a robust solution to our registration that is surprisingly simple to employ in practice. Namely, we will detail an algorithm, EOP-PM, that uses all information in the collection to conservatively determine which points on any two shapes are likely in correspondence, and show how one can adapt this to naturally get consistent registration of a wide class of surfaces. This is done by employing what we call the {\emph{approximately cycle-consistent condition}}, or ACC condition, a constraint in which registrations are only made if they are very likely with respect to those corresponding to EOP flows.

The remainder of the paper is organized as follows. In Section~\ref{sec:prelim}, we go over necessary preliminaries as well as review related work. In Section~\ref{sec:fiber}, we investigate our parallel transport model in order to derive simple constraints on which to restrict our search space. In Section~\ref{sec:registration}, we use the discussion from Section~\ref{sec:fiber} to detail a novel algorithm for registering disparate shapes given information from an entire collection, which we then adapt to compute consistent registrations in Section~\ref{sec:consistent}. We make concluding remarks in Section~\ref{sec:conclusion}

\section{Preliminaries}\label{sec:prelim}
\subsection{Registration}
We denote by $\mathcal{S}$ a collection of shapes $S_{1},...,S_{n}.$ For the purposes of our paper, we will be concerned with homeomorphic triangular meshes in $\mathbb{R}^{3},$ i.e. $S_{i} = (V_{i},F_{i}),$ where $V_{i}$ and $F_{i}$ are the set of vertices and faces respectively. Though all shapes we work with are homeomorphic to either the unit disk or the unit sphere, the ideas naturally generalize to other topologies. The astute reader will observe that many of our proposed algorithms and their implementation depend only on the point cloud structure; the mesh structure is used for initial computations as needed as well as for reparametrization as discussed in Section~\ref{sec:consistent}.

We denote by $\mathcal{C}_{\mathcal{S}}$ a collection of precomputed pairwise registrations $f_{ji}: S_{i} \to S_{j},$ with $f_{ii} = \text{Id}_{S_{i}}.$ We do not make any assumptions a priori as to the consistency or invertibility of the registrations in $\mathcal{C}_{\mathcal{S}}.$ We only assume that correspondences are everywhere defined so that an arbitrary composition is meaningful. The collection $\mathcal{C}_{\mathcal{S}}$ can be used to define notions of dissimilarities between two shapes. For the purposes of this work, we focus on dissimilarities of the following form:

\begin{equation}\label{eqn:CPD}
\begin{aligned}
d_{\mathcal{C}_{\mathcal{S}}}(S_{i},S_{j}) &= \frac{1}{2}   \left( \int_{S_{i}} g(x,f_{ji}(x))^{2} \, dS_{i} \right)^{\frac{1}{2}} \\ 
&+ \frac{1}{2} \left( \int_{S_{j}} g(x,f_{ij}(x))^{2} \, dS_{j}\right)^{\frac{1}{2}}
\end{aligned}
\end{equation}

\noindent where $g \geq 0$ is a real-valued dissimilarity, and both integrals are with respect to the usual area measures on each shape. Note that $d_{\mathcal{C}_{\mathcal{S}}}$ is nonnegative and symmetric by definition. In many cases, functionals of forms similar to that of Equation~\eqref{eqn:CPD} can be used to obtain the initial registrations in $\mathcal{C}_{\mathcal{S}}$ by an optimization procedure; see, for example, \cite{koehl2013automatic,lipman2013continuous}. To simplify notation, we will often write $d_{ij}$ in place of $d_{\mathcal{C}_{\mathcal{S}}}(S_{i},S_{j}).$ For the remainder of the work, we will assume that the $d_{\mathcal{C}_{\mathcal{S}}}$ is an actual distance function, i.e. satisfies a triangle inequality. We do not view this as a huge restriction, as one can always obtain a distance from a general class of similarity scores via diffusion maps \cite{coifman2005geometric, coifman2006diffusion}. 

With a notion of distance, we can naturally represent the collection of shapes $\mathcal{S}$ with registrations $\mathcal{C}_{\mathcal{S}}$ as a weighted complete directed graph $\Gamma(\mathcal{S},\mathcal{C}_{\mathcal{S}}),$ or $\Gamma$ for short. Each vertex $v_{i} \in \Gamma$ corresponds to the shape $S_{i}$ in our collection, and each directed edge from $v_{i}$ to  $v_{j}$ has as a weight the formal pair $(w_{ij},f_{ji}) = (K_{\sigma}(d_{ij}),f_{ji}),$ where $K_{\sigma}$ is a nonnegative kernel, such as the Gaussian kernel $K_{\sigma}(d_{ij}) = e^{- \frac{d_{ij}^{2}}{\sigma}}.$ This particular representation has been extensively studied in the bundle diffusion maps literature (e.g. \cite{singer2012vector, gao2019diffusion}) as a way to analyze datasets in which any two datapoints can be related to one another by means of an {\emph{$G$-action}}, where $G$ is usually a group ($O(2)$ in \cite{singer2012vector}) or groupoid (registrations in \cite{gao2019diffusion}). We observe that our choice of weights on directed edges naturally yields an analogous formal pair of weights for all {\emph{paths}} between any two vertices on $\Gamma.$ More precisely, an arbitrary sequence of shapes $S_{i_{0}},...,S_{i_{n}}$ naturally corresponds to the formal pair 

\begin{equation}\label{eqn:pathWt}
\left( \prod_{k=0}^{n-1} w_{(i_{k+1})i_{k}}, f_{i_{n}i_{n-1}} \circ ... \circ f_{i_{1}i_{0}} \right) 
\end{equation}

\noindent where $\circ$ denotes composition. We will often shorten this to $f_{\gamma},$ for $\gamma$ a path in $\Gamma.$

\subsection{Differential Geometry and Fiber Bundles}

To properly develop the methods proposed in this paper, we require some language from differential geometry that has not traditionally been used in applications, but has been recently proposed in \cite{gao2019diffusion}. We will model our collection of data as a {\emph{fiber bundle}}, whose definition we take from \cite{michor2008topics}.

\begin{definition}\label{def:bundle}
A {\bf{fiber bundle}} is a formal 4-tuple $(E,B,F,\pi),$ where $E, B,$ and $F$ are smooth manifolds, and $\pi: E \to B$ is a smooth mapping such that each $b \in B$ has an open neighborhood $U_{b}$ such that $\pi^{-1}(U_{b})$ is diffeomorphic to $U_{b} \times F.$ Here, $E$ is known as the {\emph{total manifold}}, $B$ is known as the {\emph{base manifold}}, and $F$ is known as the {\emph{fiber}}.
\end{definition}

In other words, a fiber bundle is locally a product manifold. Every $b \in B$ has a corresponding fiber $F_{b}$ that is diffeomorphic to $F.$ In our setting, we model our collection of shapes as fibers in a fiber bundle , with shape $S_{i}$ having corresponding point $b_{i}$ in the base manifold (analogous to their corresponding vertices $v_{i}$ in $\Gamma.$) The distance between $b_{i}$ and $b_{j}$ is thus given by $d_{ij}.$ 

The fiber bundle framework allows us to model registrations as {\emph{parallel transports}}, which lie at the heart of our method. In Riemannian geometry, a smooth curve $\gamma$ between two points $p$ and $q$ on a Riemannian manifold $M$ naturally induces a parallel transport operator $P_{\gamma}: T_{p}M \to T_{q}M.$ Parallel transports in this setting have been used in medical imaging \cite{lorenzi2013geodesics} to align vector field deformations of images to a common template. Our use of meshes does not allow for any direct adaptation of such methods to our setting, e.g. fiber bundles. We reserve the discussion of full details for the Appendix. The key fact that we require is a basic consequence of the defintion, which we summarize here: parallel transports on the base manifold naturally induce diffeomorphisms of the fibers \cite{michor2008topics}. This leads us to our ability to mathematically model computed correspondences as parallel transports. 

\begin{figure}[t]
\begin{center}
\centering
\includegraphics[width = 0.3\textwidth]{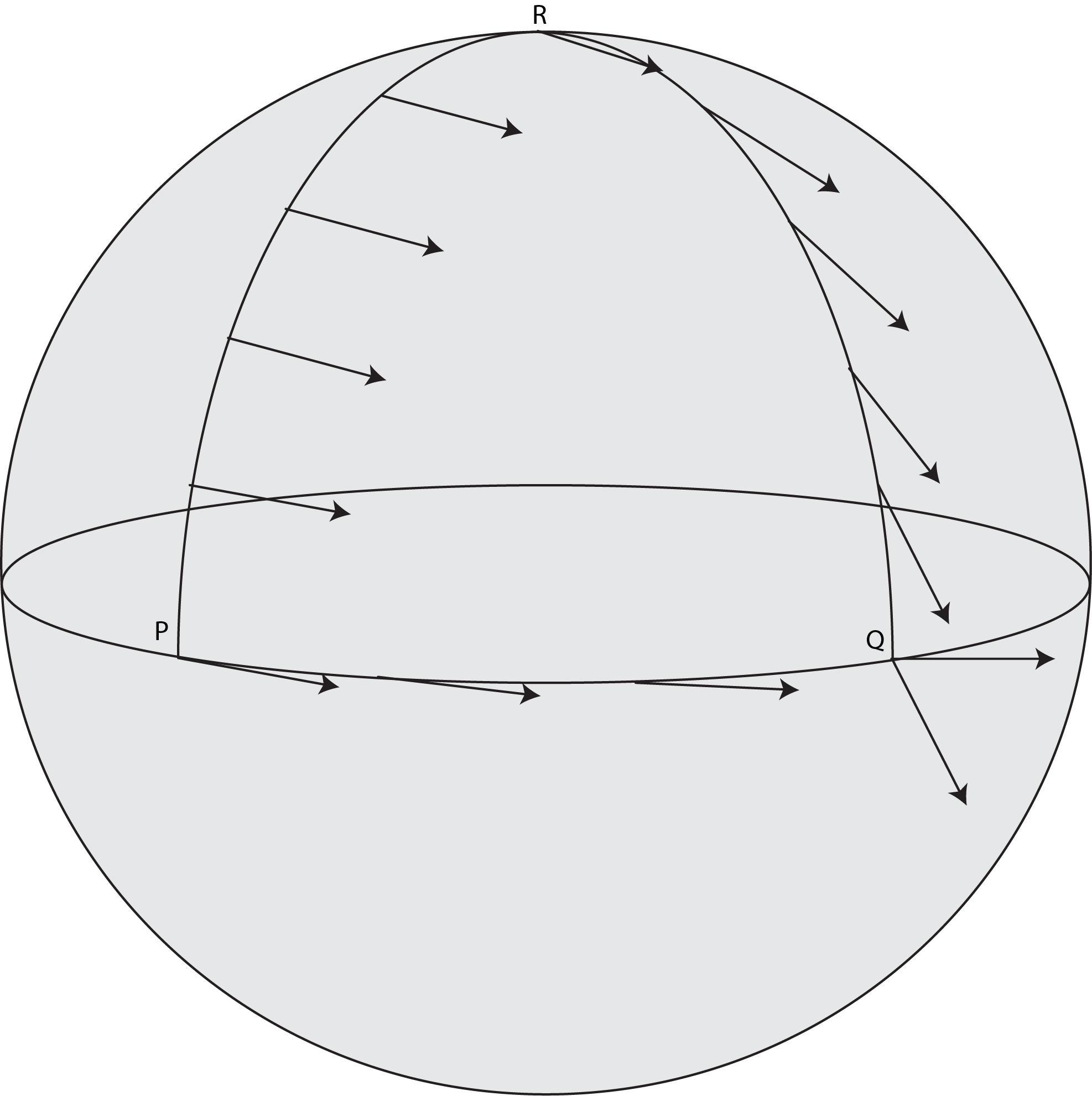}
\end{center}
\caption{Two parallel transports of a tangent vector at a point $P$ on a manifold. One parallel transport corresponds to the geodesic between $P$ and $Q,$ whereas the other corresponds to composing the geodesic from $P$ to $R$ with that of $R$ to $Q.$}\label{fig:ParallelTransport}
\end{figure}

The utility of this modeling mechanism can be seen through {\emph{holonomy}}. The precise definitions are out of scope for this work, but we illustrate the intuitive idea in Figure~\ref{fig:ParallelTransport}, which shows an example of the Riemannian manifold setting; we refer the interested reader to recent work on synchronization for more thorough discussions\cite{gao2019geometry}. Here, we see that different paths result in different parallel transports of tangent vectors. Since parallel transports are registrations in our fiber bundle setting, this corresponds to inconsistency of a collection of registrations. A natural question one can ask is whether it is possible to determine a collection of paths for which the parallel transports are reasonably close to one another; this is central to our methodology. 
%
\section{Finding Meaningful Paths for Parallel Transport}\label{sec:fiber}
Though our ultimate interest is in obtaining consistent registrations, we first concern ourself with registration quality. For each pair of shapes $(S_{i},S_{j})$ our assumed collection of registrations naturally yield an infinite number of registrations between $S_{i}$ and $S_{j}.$ Indeed, every composition of the form $f_{jk_{n}} \circ ... \circ f_{k_{0}i}$ is a valid registration. It is thus natural to ask if there is any proxy to estimate the quality of such registrations without requiring ground truth information. In this section, we detail such a condition motivated by parallel transport.

We first immediately note that the infinite number of registrations can be trimmed to a finite (but possibly large) subset without issue. Indeed, since each such registration corresponds to a path $\gamma$ in $\Gamma,$ we can consider two different classes of paths: ones which contain loops, and ones which do not. The collection of paths that do not contain loops is finite as $\Gamma$ only has a finite number of vertices, meaning that the collection containing loops is infinite. Recall that we assumed that the collection of registrations is not consistent. Thus, if $\tilde{\gamma}$ is a loop in $\Gamma,$ the corresponding registration $f_{\tilde{\gamma}}$ is not necessarily the identity map; if the collection was consistent, then $f_{\tilde{\gamma}}$ would necessarily be the identity. This means that loops function only as error perturbations; each loop only adds error to a registration corresponding to a simple path without loops. 

We can thus safely remove all registrations with loops from consideration, leaving only a finite set left. Nevertheless, this finite set is potentially large as the number of simple paths in a graph grows factorially with the number of vertices, potentially making it prohibitively large to work with. Though one could proceed via sampling, it is natural to ask if we can further reduce the collection under consideration. A common assumption within the general area of shape analysis is that good registrations can be obtained from geodesics in shape space \cite{klassen2004analysis, lorenzi2013geodesics}. If we keep this assumption, then a reasonable way forward is to restrict the collection of paths of interest to those whose parallel transports are close to those of geodesics. A particular way of measuring this is given by the following result, which is an exercise in \cite{chavel2006riemannian}.

\begin{theorem} \label{thm:EOPEstimate} Let $H(s,t)$ be an arbitrary homotopy between two smooth curves $\gamma_{0}(t):=H(0,t)$ and $\gamma_{1}(t):=H(1,t)$ that do not intersect except at endpoints $p:=H(s,0)$ and $q:=H(s,1).$ Let $X \in T_{p}M$ be a unit vector, and let $X_{s}$ be the vector field defined by parallel transporting $X$ along $\gamma_{s}(t):=H(s,t)$ for $s$ fixed. Then

\begin{equation}\label{eqn:EOPEstimate}
||X_{1}(1)-X_{0}(1)|| \leq \frac{4}{3} K_{max}A(H)
\end{equation}

\noindent where $A(H)$ is the area of the graph of the homotopy and $K_{max}$ is the maximum of the absolute Gaussian curvature of the graph of the homotopy.
\end{theorem}

In other words, the difference between the parallel transport of a geodesic and a curve homotopic to a geodesic is a function of the curvature of the manifold and the area of \emph{any} homotopy. Though we do not necessarily have control over the curvature of a given shape space, it is possible using metric information to control the area term. To this end, we make a definition.

\begin{definition} \label{def:eyesOnPrize}
For vertices $v_{i}$ and $v_{j}$ in a connected weighted graph $\Gamma = (V,E)$ with edge weights given by a matrix $D = (d_{ij})$ with $d_{ij} = d(v_{i},v_{j})$ for some metric $d,$ we define the {\bf{Eyes on the Prize (EOP) matrix between $v_{i}$ and $v_{j}$}}, denoted by $F_{i \to j}^{\Gamma},$ is the $|V| \times |V|$ binary matrix defined by

\begin{equation}
(F_{i \to j}^{\Gamma})_{mn}:=  \begin{cases} 1 & \text{if } d_{ij} < d_{in} \emph{\text{ and }} d_{jm} > d_{jn} \\0 & \emph{\text{otherwise}} \end{cases}
\end{equation}
\end{definition} 

\noindent The EOP name comes purely from intuition: a path on this adjacency matrix is simultaneously moving towards its target and can never move back towards its source, e.g. ''its eyes are always on the prize." The directed paths corresponding to the adjacency matrix $F_{i \to j}^{\Gamma}$ are those that both always move away from $v_{i}$ and always move towards $v_{j}.$ We will refer to these paths as {\emph{EOP directed flows}}.From the perspective of Theorem~\ref{thm:EOPEstimate}, the efficient movement of the paths of $F_{i \to j}^{\Gamma}$ limits the area of the minimizing homotopy of Equation~\eqref{eqn:EOPEstimate}.

We conclude this section by briefly remarking as to potential theoretical guarantees of paths from the $F_{i \to j}^{\Gamma}.$ Given the minimal information assumed, it is difficult to derive explicit estimates of Equation~\eqref{eqn:EOPEstimate}, as one requires information about the underlying shape space that is not necessarily available a priori in order to derive such estimates. Nevertheless, it is possible to conclude the following:

\begin{theorem}\label{thm:EOPFunction}
Let $x$ and $y$ be two points in $M$ with minimizing geodesic $\gamma.$ Let $C_{x,y}^{\varepsilon}$ be the collection of unit speed $C^{1}$ curves $s$ in $M$ that satisfy the following conditions

\begin{align*}
g \left(s'(t),v_{x,s(t)} \right)  > \varepsilon \\
g \left(s'(t),v_{s(t),y} \right) > \varepsilon
\end{align*}

\noindent where $1 > \varepsilon > 0$ and $v_{x,s(t)}$ and $v_{s(t),y}$ are the initial condition vectors for the unit speed geodesic between $x$ and $s(t),$ and $s(t)$ and $y$ respectively.  Then, given a Fermi coordinate neighborhood of $\gamma,$ there is a subneighborhood such that every path in $C_{x,y}^{\varepsilon}$ whose image is contained in the subneighborhood is necessarily a function of the coordinate in the neighborhood corresponding to the geodesic.
\end{theorem}

\noindent All details are deferred to the Appendix. To summarize, there is a choice of coordinates for which the more efficient paths in $F_{i \to j}^{\Gamma}$ can necessarily be written as functions of one coordinate. Given a specific model of a shape space, one can then derive estimates of the bounds in Equation~\eqref{eqn:EOPEstimate}. As we are explicitly attempting to not reference a particular model of shape space, we defer such discussions for future work. 

\section{Improving Correspondence Quality via EOP Flows}\label{sec:registration}
\begin{algorithm} 
\caption{Eyes on the Prize Soft Registration (EOP-SR)}
\label{alg:softCorrespondenceAlg}
\begin{algorithmic}

\State {Input: Shapes $S_{1},...,S_{k}$; distances $(d_{ji})$; registrations $\{f_{mn}\}$; threshold $\lambda$}
\State {Compute: $W:= (w_{mn}) = (K(d_{mn}))$ }
\For{$i \in \{1,...,k\}$} 
	\For{$j \in \{1,...,k\}$}
		\If{$i \neq j$} 
			\State{Compute $F_{i \to j}^{\Gamma}$}
			\State{$WF_{i \to j}^{\Gamma} = W \odot F_{i \to j}^{\Gamma}$}
			\State{$f^{s}_{ji} =$ Modified\_DFS$(\{f_{mn}\},WF_{i \to j}^{\Gamma},\lambda)$}
		\EndIf
	\EndFor
\EndFor
\State {Output: Soft correspondences $\{f^{s}_{mn}\}$}

\end{algorithmic}
\end{algorithm}
Algorithm~\ref{alg:softCorrespondenceAlg} details EOP-SR, one of the main contributions of our paper. The algorithm seeks to distill \emph{soft registrations} from an initial collection of registrations and distances \cite{solomon2012soft}. For each pair of shapes, the EOP matrix from Definition~\ref{def:eyesOnPrize} is first computed. From that, we construct the corresponding a weighted version of the EOP matrix, done by taking the Hadamard (pointwise) product with the weight matrix $W = (w_{ij}),$ where $w_{ij} = K(d_{ij})$ for some kernel function $K.$ The soft registrations are then constructed via a depth first search approach that keeps track of both the composed correspondences as well as the weights. Recall from Equation\~eqref{eqn:pathWt} that the weight of a path is the product of the weights of the edges. The hard threshold $\lambda$ is user-defined: such thresholds could include the number of edges used in a given path (where paths with too many edges are omitted), or path weight (where paths with low weight are omitted). In practice, the depth first search procedure has an iterative relaxation procedure. If the minimum path weight specified is too high, i.e. no paths have the minimum weight, the search is repeated with a gradually decreasing threshold until a specified number of paths are admitted.

The resulting output is thus a collection of probability distributions: each point in $S_{i}$ has a corresponding probability distribution of points in $S_{j}.$ Our use of soft correspondences naturally leads to robustness; small changes in the number of surfaces considered will naturally lead to small changes in the soft correspondences. In order to obtain an actual correspondence from the soft correspondences computed above, one need only to take a statistic of the soft correspondences. We restrict our attention to the maximum likelihood estimate and Frechet mean with respect to geodesic distance. 

\subsection{Accuracy of Ground Truth Propagation}

\begin{figure*}
\begin{center}
\centering
\includegraphics[width = 0.9\textwidth]{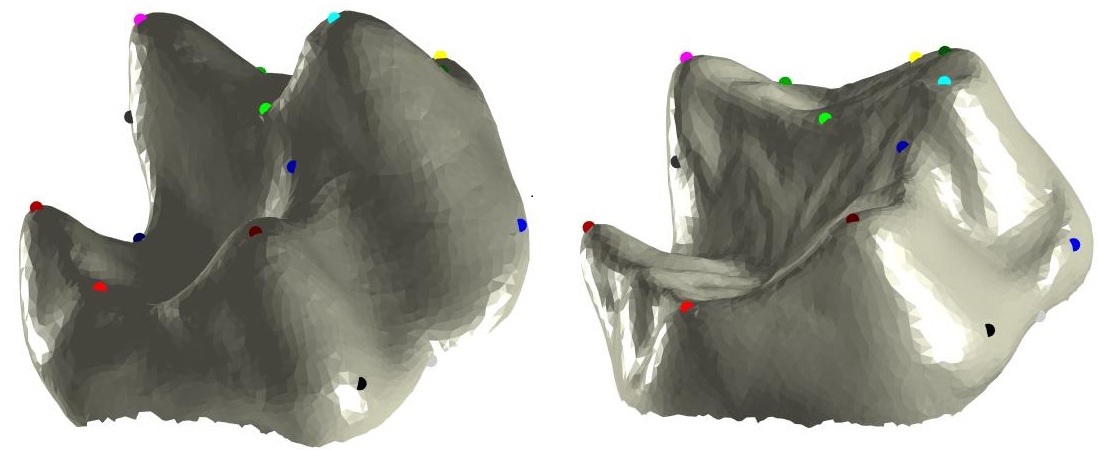}
\caption{Two different surfaces with observer landmarks taken from \cite{boyer2011algorithms}. Corresponding landmarks have the same color.}\label{fig:LandmarkExample}
\end{center}
\end{figure*}

To evaluate the quality of the soft registrations computed by Algorithm~\ref{alg:softCorrespondenceAlg}, we test our ability to propagate given ground truth for a collection of anatomical surfaces. We use as data the 116 tooth crowns used in \cite{boyer2011algorithms}. These teeth are equipped with a consistent set of eighteen observer landmarks divided into two categories: type 2, which correspond to nongeometric features relative to the position of the tooth in the jaw, and type 3 landmarks, which corresponding to more familiar geometric features such as cusps and saddles. Two such landmarked surfaces are given in Figure~\ref{fig:LandmarkExample}. 

For each pair of surfaces, we compute two different initial collections of registrations and induced distances via two unsupervised methods: those via the continuous Procrustes distance \cite{lipman2013continuous}, and those via the GP-BD procedure defined in \cite{gao2019gaussian}. Registrations from the continuous Procrustes method are constructed by first computing local curvature extrema, finding an optimal M\"{o}bius transform via exhaustive search that minimizes distances between these extrema after the surfaces are embedded in the hyperbolic disk, and finally interpolates the correspondences via a thin plate spline. The GP-BD methodology focuses on sampling a number of Gaussian Process landmarks, for which initial putative matches are then made via the wave kernel signature \cite{aubry2011wave} and then subsequently refined and interpolated via methods optimizing bounded distortion \cite{lipman2014feature, kovalsky2015large, kovalsky2016accelerated}. Note that both methods neither require knowledge of ground truth nor significant preprocessing such as alignment; all correspondences are learned. For both collections, we consider the following five methods for landmark propagation

\begin{itemize}
\item Direct pairwise propagation
\item Propagation along a minimum spanning tree
\item Propagation along the shortest path in the 5-NN graph
\item Maximal likelihood estimate of the output of EOP-SR (EOP Mode)
\item Frechet mean of the output of EOP-SR (EOP Mean)
\end{itemize}

The first two methods listed are currently utilized and are present to serve as a benchmark as well as to illustrate the general inaccuracies that can occur by solely relying on direct propagation. The third method is natural to consider given the context of the paper, as such a path is an approximation of the geodesic in shape space. The last methods are statistics derived from EOP-SR. 

The kernel we use in defining our weight matrices is the Gaussian kernel $K(d) = e^{-d^{2}/\sigma}$. Letting $\bar{d}$ denote average of the distances between each pair of surfaces, the default parameters used for the correspondences derived from Algorithm~\ref{alg:softCorrespondenceAlg} are as follows: we let $\sigma = 0.25 \bar{d},$ the maximum number of edges in any path is 4, and the minimum weight of any path is $e^{-\bar{d}^{2}}.$

\begin{figure*}[t]
\begin{center}
\centering

\includegraphics[width=.45\textwidth]{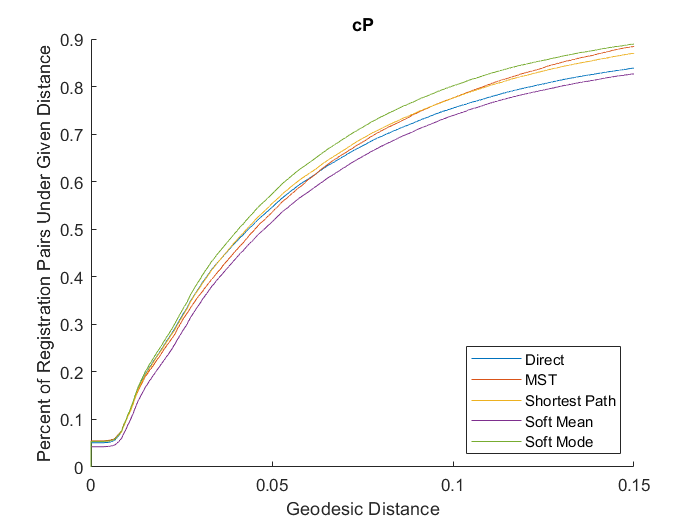}
\hfill
\includegraphics[width=.45\textwidth]{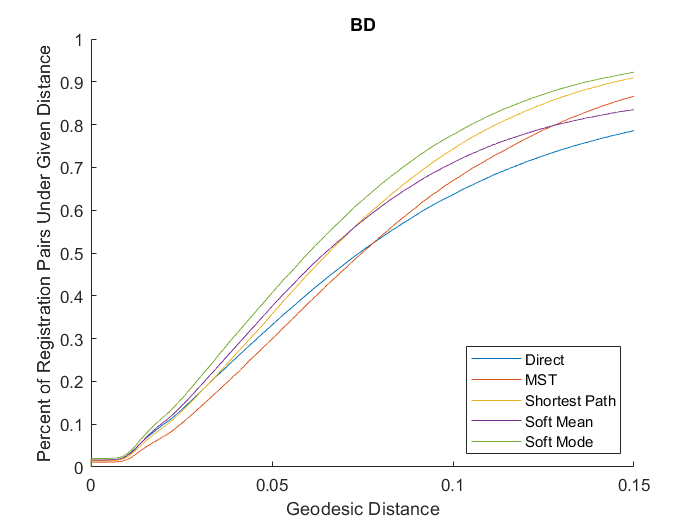}
	\captionof{figure}{Curves detailing the accuracy of landmark propagation for each of the three initial collections of registrations computed. Left: Continuous Procrustes maps. Right: GP-BD Maps.}\label{fig:EOPBenchmarkTotal}
\end{center}
\end{figure*}

To visualize the accuracy of our method, we utilize the Princeton benchmark introduced in \cite{lipman2009mobius}. The axes of this graph are geodesic distance and percentage of computed landmarks that are at most the corresponding geodesic distance away from the ground truth. Each curve on this graph represents the resulting accuracy of the correspondences. We see the results for all pairwise correspondences for all pairs of surfaces in Figure~\ref{fig:EOPBenchmarkTotal}. We see that for both collections of registrations, propagation via the mode of our soft correspondences is more accurate than the other methods utilized. That the mode is generally more accurate than both the direct and MST-based propagation is not surprising; the former tends to fail for highly nonisometric pairs, and our soft correspondences do not have any consistency requirements as imposed by the MST-based method. It is perhaps more surprising, however, that the soft correspondence mode is generally superior to the shortest path method, and that the Frechet mean does as poorly as it does. To the former, we remark that any shortest path computed is one particular approximation to a geodesic, while our soft correspondences essentially incorporate multiple, allowing for more robust decision making. Given the success of the soft correspondence mode, the failure of the soft correspondence Frechet mean is likely due to the shape of the support of the distribution.

\subsection{Evaluation of Parameter Choices}

\begin{figure*}[t]
\begin{center}
\centering

\includegraphics[width=.45\textwidth]{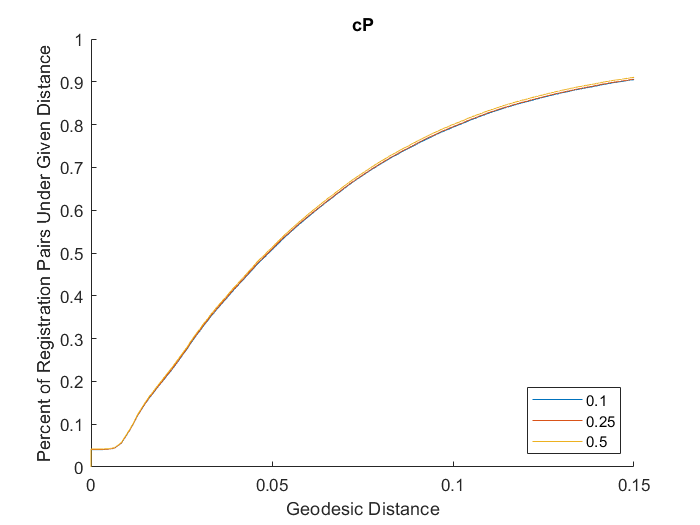}
\includegraphics[width=.45\textwidth]{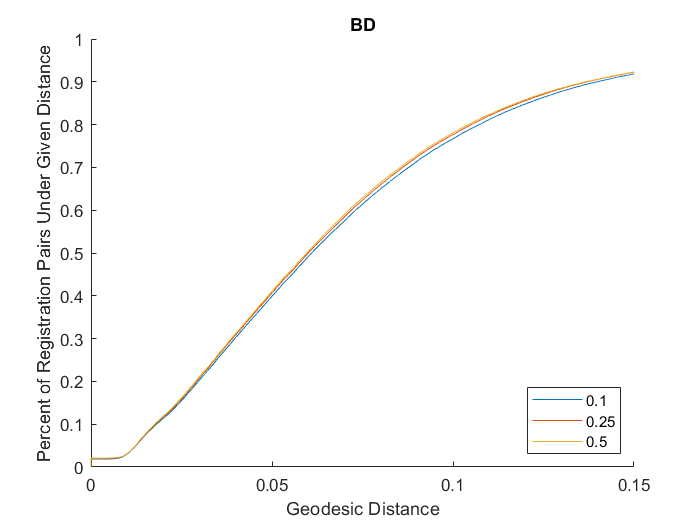}
\\
\includegraphics[width=.45\textwidth]{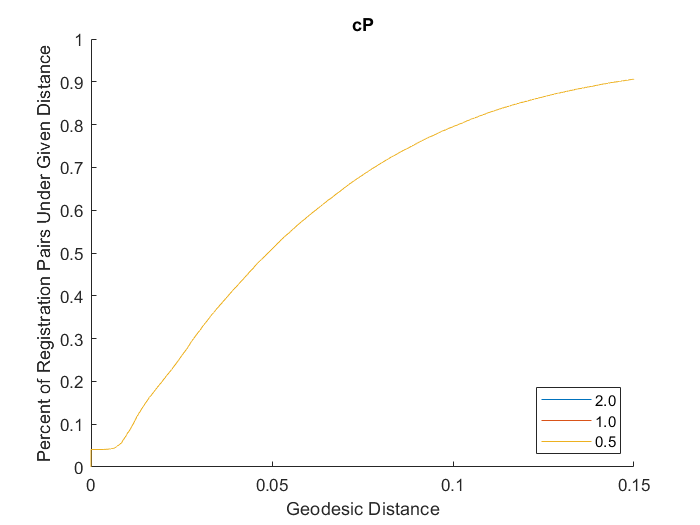}
\includegraphics[width=.45\textwidth]{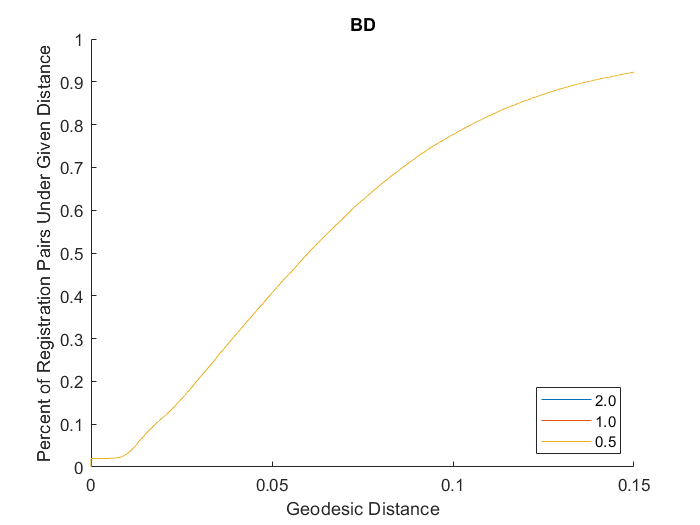}
\\
\includegraphics[width=.45\textwidth]{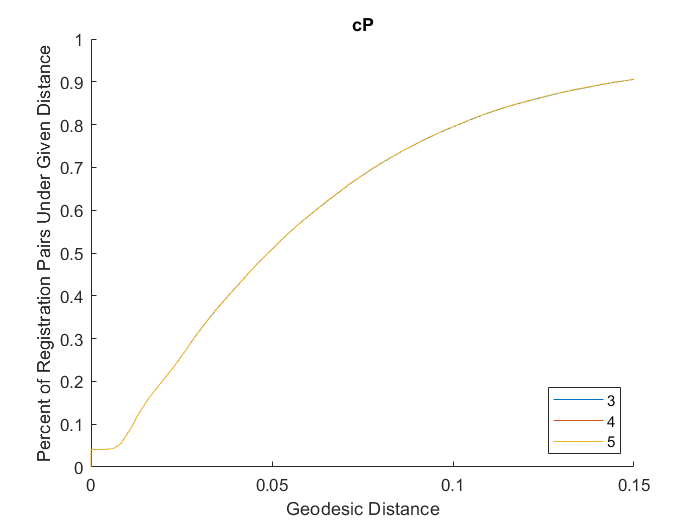}
\includegraphics[width=.45\textwidth]{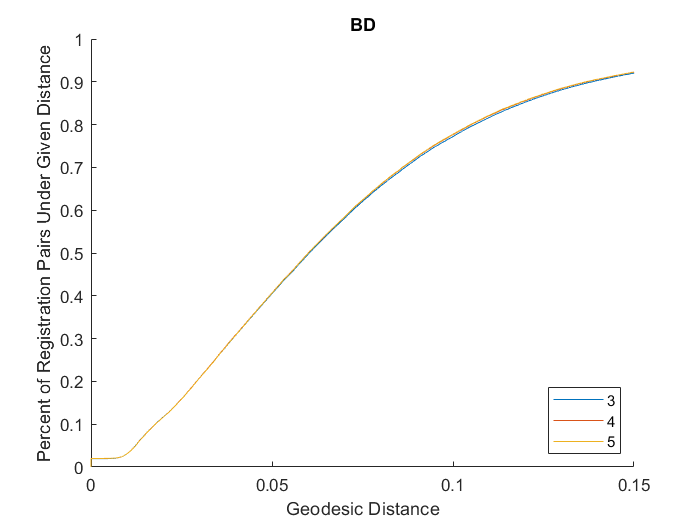}

	\captionof{figure}{Ablation Experiments. Top: Effect of changing the temperature parameter ($\sigma$). Middle: Effect of changing the minimum admissible path weight. Bottom: Effect of changing the maximum number of edges considered.}\label{fig:EOPAblation}
\end{center}
\end{figure*}

In the defintion of Algorithm~\ref{alg:softCorrespondenceAlg}, a number of different parameters choices are mentioned. It is naturally interesting to see what, if any effect, they have on the accuracy results. Figure~\ref{fig:EOPAblation} shows the accuracy results for the soft correspondence mode propagation when changing three separate parameters. The top of the figure considers three different values of $\sigma = a \bar{\sigma},$ where $a = 0.1, 0.25, 0.5.$ The middle considers changes in admissible path weight
of the form $b e^{-\bar{d}^{2}},$ where $b = 0.5, 1, 2.$ Lastly, the bottom of the figure considers changes in the maximum number of edges in any admissible path, ranging between 3 and 5. The figure shows negligible change in all cases, suggesting robustness of the parameters chosen.

\section{Application: Robust and Consistent Surface Registration}\label{sec:consistent}
The results from the previous section suggest that our forward propagation methodology can significantly improve registration quality. Those results, however, are alone insufficient for real-world application for two reasons. First, any rigorous analysis, such as one using standard shape space methods \cite{kendall1984shape, kendall1989survey, dryden2016statistical}, ultimately requires consistent registration of all surfaces in a collection. Both Algorithm~\ref{alg:softCorrespondenceAlg} and all experiments related to it conducted in Section~\ref{sec:registration} make no assumption about end result concerning consistency; the algorithm makes a choice based on the computed soft correspondence, but has no requirements on consistency, nor is there any requirement that the points used in each surface belong to some canonical subset such as the landmarks from \cite{boyer2011algorithms}. The algorithm could be easily modified if such a subset were known: for example, a projection operation could have been added to the experiments in Section~\ref{sec:registration} that would require any decision be in terms of one of the eighteen landmarks used. However, the second issue we bring up is that such an expectation of always having a canonical set of expert landmarks is not always feasible in practice, especially for large collections of surfaces. Even if they were to exist, we would be immediately presented with a paradox: from a practical perspective, if we knew such landmarks a priori, we would also likely know how they are registered! With such registrations, we could then easily obtain registrations of the entire surface, assuming that such a surface is a topological disk or sphere, using established methods for interpolating sparse correspondences \cite{kovalsky2015large, kovalsky2016accelerated, aigerman2016hyperbolic}. Thus, the main practical bottlenecks are to first determine sparse correspondences between a given pair of surfaces which can be interpolated to yield a high-quality registration, and then enforce that such registrations are consistent.

\subsection{Robust Partial Matching of Gaussian Process Landmarks}

\begin{figure*}
\begin{center}
\centering
	\includegraphics[width=\textwidth]{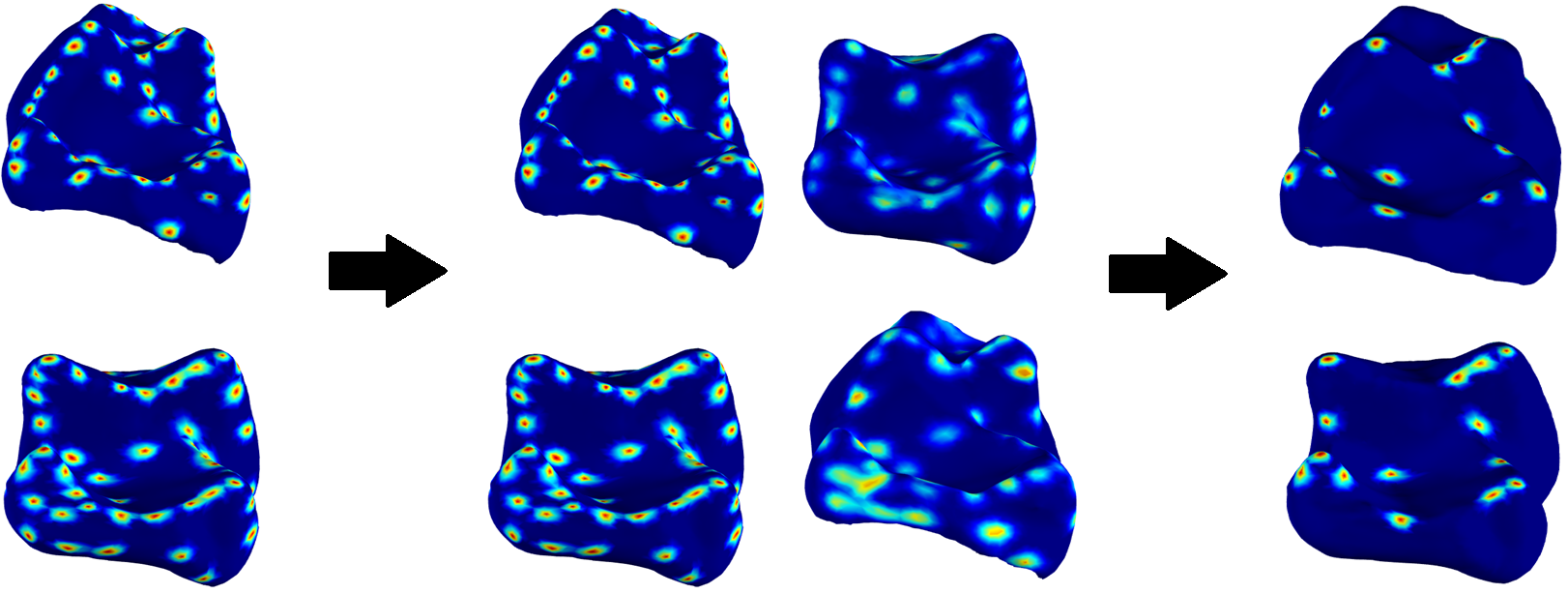}
\end{center}
\caption{A visual summary of our collection-based correspondence refinement procedure. After choosing salient landmarks on each shape, a subset of an initial collection of registrations of potentially variable quality are used to determine likelihoods of correspondences for each landmark. Final landmarks are chosen by a cycle-consistency based affinity.}\label{fig:EOPProcedure}
\end{figure*}

Maps generated by the continuous Procrustes (cP) \cite{lipman2013continuous} and Gaussian Process-Bounded Distortion (GP-BD) \cite{gao2019gaussian} procedure operate at a high level in the same way. First, candidate points on each surface of interest are generated, matches between said candidate points (if they exist) are obtained, and the resulting sparse correspondences are then interpolated in order to complete the registration procedure. In particular, not all of the generated points are placed in correspondence; the output of each method prior to interpolation is a {\emph{partial}} matching with some of the generated candidate points omitted. Given the success of our soft correspondence mode in improving registration accuracy, as shown in Figure~\ref{fig:EOPBenchmarkTotal}, we would like to adapt our procedure in order to generate robust sparse correspondences of surfaces using only the prior information given, i.e. namely the surfaces and the initial collection of registrations and distances.

\begin{algorithm} 
\caption{Eyes on the Prize Partial Matching (EOP-PM)}
\label{alg:partialMatchingAlg}
\begin{algorithmic}

\State {Input: Shapes $S_{a},S_{b}$; distances $(d_{ji})$; registrations $\{f_{ji}\}$; \\ threshold $\lambda$, possible landmarks $\{v_{1}^{a},...v_{k_{A}}^{a}\} = A \subseteq S_{a}$ and $\{v_{1}^{a},...v_{k_{B}}^{b}\}= B \subseteq S_{b},$ error radius $\varepsilon > 0$}
\State{Matches = []}
\For{$i \in \{1,...,k_{A}\}$}
	\State{Compute $f_{ba}^{s}(v_{i}^{a})$ from Algorithm~\ref{alg:softCorrespondenceAlg}}
\EndFor
\For{$i \in \{1,...,k_{B}\}$}
	\State{Compute $f_{ab}^{s}(v_{i}^{b})$ from Algorithm~\ref{alg:softCorrespondenceAlg}}
\EndFor
\For{$i \in \{1,...,k_{A}\}$}
	\For{$j \in \{1,...,k_{B}\}$}
		\If{$\mathbb{P}(f_{ba}^{s}(v_{i}^{a}) \in B_{\varepsilon}(v_{j}^{b})) > 0.5$}
			\If{$\mathbb{P}(f_{ab}^{s}(v_{j}^{b}) \in B_{\varepsilon}(v_{i}^{a})) > 0.5$}
				\State{Matches.append($(v_{i}^{a},v_{j}^{b})$)}
				\State{BREAK}
			\EndIf
		\EndIf
	\EndFor
\EndFor

\State {Output: Matches}

\end{algorithmic}
\end{algorithm}

Our proposed partial matching algorithm can be found in Algorithm~\ref{alg:partialMatchingAlg}, for which a demonstrated workflow can be found in Figure~\ref{fig:EOPProcedure} using registrations from the cP procedure on a pair of surfaces from \cite{boyer2011algorithms} using 50 Gaussian process landmarks\cite{gao2019gaussianMain, gao2019gaussian}, which have been empirically shown to generate a collection of landmarks used that outlines regions of significant curvature on a surface. First, for a given pair of shapes $S_{a}$ and $S_{b},$  we follow Algorithm~\ref{alg:softCorrespondenceAlg} to generate soft correspondences for each of the possible landmarks in $A$ and $B$. We then employ a conservative mutual matching procedure. A pair of landmarks $(v_{i}^{a},v_{j}^{b})$ are matched if most of the probability mass of $f_{ba}^{s}(v_{i}^{a})$ is concentrated around a small geodesic ball $ B_{\varepsilon}(v_{j}^{b}))$ around $v_{j}^{b},$ where $\varepsilon>0$ is a user-input quantity, and vice versa. Given that the soft correspondences are generated from a family of efficient paths between $S_{a}$ and $S_{b},$ we are in essence declaring that a pair of landmarks are registered if their soft correspondences are {\emph{approximately cycle-consistent}} (ACC) with respect to the family of paths chosen. Note in particular that, for $\varepsilon$ sufficiently small, the ACC constraint guarantees that each landmark on $S_{a}$ can be paired with at most one landmark on $S_{b}$ and vice versa, removing any possibility of tie-breaking. Note that Figure~\ref{fig:EOPProcedure} shows the effect of this conservative constraint: the back cusp on the top specimen is not registered to any point on the cuspless specimen on the bottom. We do not, however, view this as a flaw in light of the expert registered landmarks given in \cite{boyer2011algorithms}. Algorithm~\ref{alg:partialMatchingAlg} is conservative by design, and only registers pairs based on their satisfaction of a relatively stringent condition. The lack of registration can only indicate sufficient uncertainty, given the initial collection of distances and registrations, to register areas of distinct nonisometry.

As a brief aside, it is perhaps of mathematical interest to note that our sparse registration procedure does not involve the direct optimization of a particular objective function. Rather, our procedure is one of determining which points satisfy the ACC constraint given by the two if-statements, and not just choosing the most cycle-consistent pairs via minimization of an objective function as in \cite{huang2013consistent}. We suspect that this distinction is important and is worth studying of its own accord, especially given the recently established importance of cycle consistency outside of registration, e.g. CycleGAN \cite{zhu2017unpaired}. Such a study is, however, out of scope for this paper and is deferred to future, much more general work.

\begin{figure*}
\begin{center}
\centering
\includegraphics[width=\textwidth]{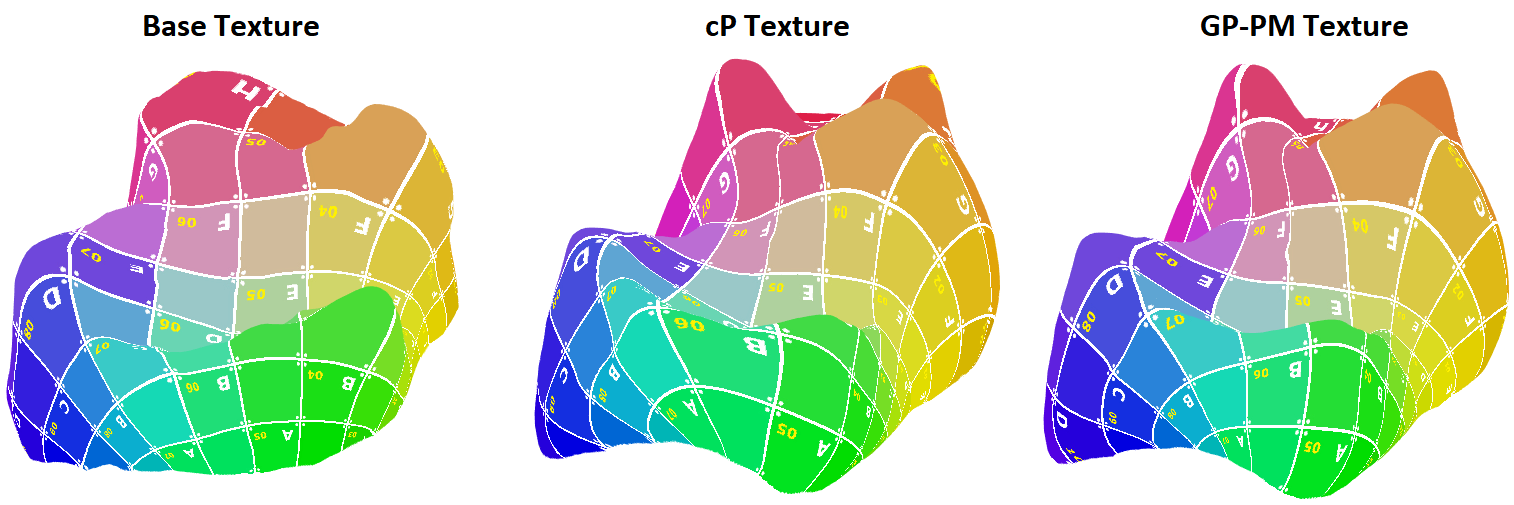}
\includegraphics[width=\textwidth]{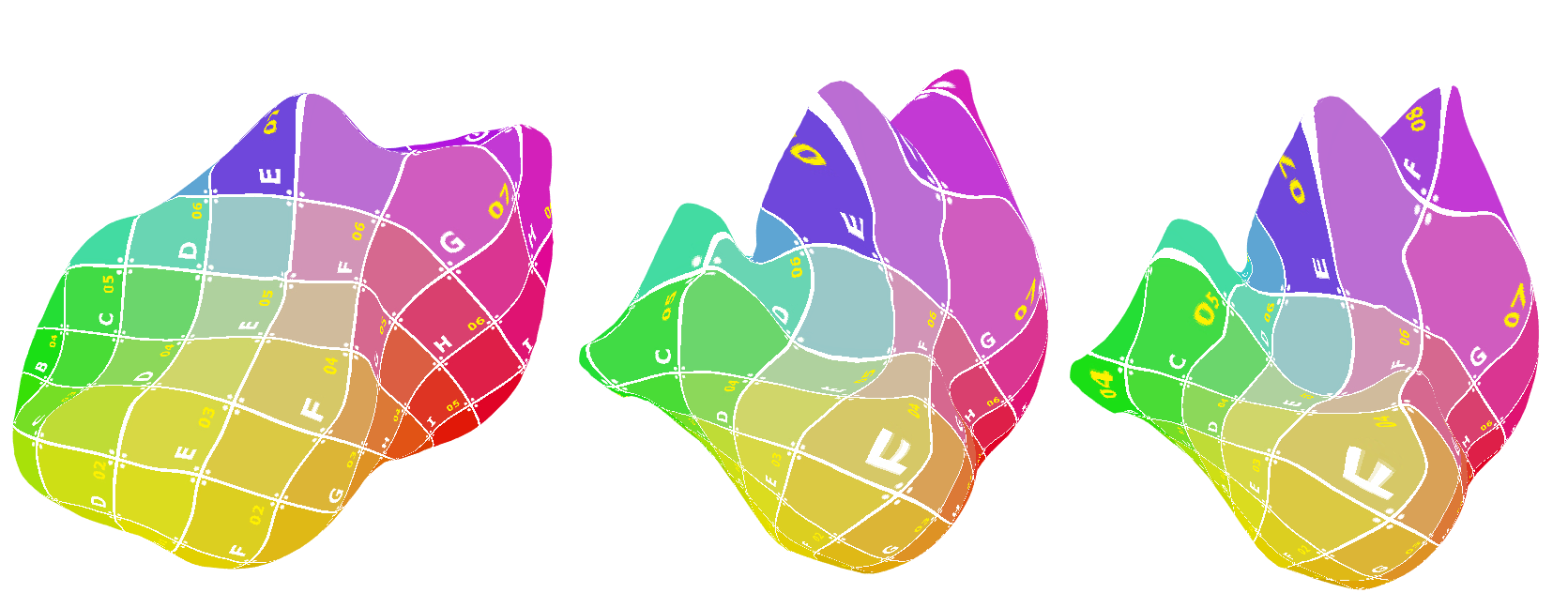}
	\captionof{figure}{Examples illustrating the ability for interpolation of partial matching to fix inaccurate correspondences with Gaussian Process landmarks. Left: tooth crown with prescribed texture coordinates. Middle: Computed continuous Procrustes map to the left tooth. Left: Map computed by interpolating Gaussian Process partial matching. Top: Fixed correspondence of ridges, Bottom: lower distortion for highly different geometries.}\label{fig:MapExamples}
\end{center}
\end{figure*}

Figure \ref{fig:MapExamples} shows two examples, one in each row, of scenarios commonly encountered in collections of highly nonisometric surfaces \cite{gao2018development}, again using surfaces from \cite{boyer2011algorithms} and Gaussian process landmarks. The left shows a surface that has been textured based on its conformal parametrization to the plane via the mid-edge uniformization \cite{lipman2009mobius}. The middle shows another surface with transferred texture from the left based on its cP registration. The right is the resulting texture transfer based on a combination of the output of Algorithm~\ref{alg:partialMatchingAlg} with interpolation based on thin plate splines, as in the cP registration method. In the top row, we note that the texture based on the cP registration is not aligned with that of the base texture; this can be seen by looking at the change in the texture along the bottom ridge. The texture obtained from the Gaussian Process partial matching procedure does not suffer from this issue. In the bottom row, we see that the pink portion of the texture on the Gaussian Process partial matching registration (around the top right of each surface) exhibits significantly less warping than that of the cP registration. Thus, our partial matching procedure using an entire collection of registrations along with the ACC constraint can allow for higher quality registration than those obtained by directly minimizing an objective, making it easy to fix these common errors.

\subsection{An Algorithm to Consistently Register Aligned Surfaces}
We are now ready to define our method for obtaining consistent registration of surfaces, but first add one additional motivating remark. Many methods used in end-to-end registration of collections of surfaces, either in their proposal or their implementation, make explicit or implicit assumptions as to the topology of surfaces in question. For example, the implementation of both the cP and GP-BD registration methods require surfaces to be topological disks. This is a major bottleneck, as many surfaces of practical interest need not be simply connected. Some methods, however, do not: the hyperbolic orbifold procedure for interpolating sparse correspondences proposed in \cite{aigerman2016hyperbolic}, for example, is theoretically valid for surfaces of higher genus (though to the best of our knowledge, no such implementation exists). Given that the methods by design depend on an initial collection of registrations, it is thus crucial that any pipeline we propose be theoretically amenable to surfaces regardless of topology; merely relying on the previously used methods depending on disk topology is thus insufficient.

Note, however, that none of our algorithms make any assumptions on topology. Indeed, the methods we have developed readily adapt to any particular choice of registration. Given our restricted interest in biological surfaces, which are generally rigid in nature, a natural candidate input to our methods appears: projection after alignment. Indeed, though biological surfaces such as the ones considered here have proven difficult to fully register, there are well established algorithms (such as Auto3DGM \cite{puente2013distances, boyer2015new}) that will accurately align large collections of such surfaces independent of assumed topology. Thus, distances and registrations from these algorithms are natural to feed in to the procedures we've outlined.
\begin{algorithm}
\caption{Obtaining Mappings by Revising Alignments (EOP-Proj)}\label{alg:MappingsFromAlignment}
\begin{algorithmic}

\State {Input: Aligned shapes $S_{1},...,S_{n}$; distances $(d_{ji})$; projections $\{f_{ji}\}$; threshold $\lambda$, number of Gaussian process landmarks $k$, mismatch thresholds $\varepsilon_{1}, \varepsilon_{2} > 0,$ maximum number of matches $m.$}
\State{AllMatches $= []$}
\State{Compute $j_{Fr} = \text{argmin}_{1 \leq i \leq n} \sum_{j=1}^{n} d_{ji}^{2}$}
\For{$i \in \{1,...,n\}$} 
	\State{Compute curvature local maxima $P^{GC}_{i}$ of $S_{i}$}
	\State{Compute $k$ Gaussian Process landmarks $P_{i}^{GP}$ of $S_{i}$}
\EndFor

\For{$i \in \{1,...,n\}$} 
	\State{CurMatches $= [].$}
	\If{$i \neq j_{Fr}$}
		\State{Compute $MP_{i}^{C}$ = EOP-PM($P^{C}_{i},P^{GC}_{j_{Fr}},\lambda,\varepsilon_{1}$)}
		\State{Compute $MP_{i}^{GP}$ = EOP-PM($P^{GP}_{i},P^{GP}_{j_{Fr}},\lambda,\varepsilon_{2}$)}
		\State{CurMatches $= MP_{i}^{C}$}
		\While{size(CurMatches) $< m$}
			\State{Compute $q = \text{argmax}_{p \in MP^{GP}_{i}} E(p,\text{CurMatches})$}
			\State{CurMatches.append($q$)}
			\State{$P^{GP}_{i}$.remove($q$)}
		\EndWhile
	\EndIf
	\State{AllMatches.append(CurMatches)}
\EndFor

\State{$\tilde{S}_{j_{Fr}} = S_{j_{Fr}}.$}
\For{$i \in \{1,...,n\}$} 
	\If{$i \neq j_{Fr}$}
		\State{Compute $\tilde{f}_{j_{Fr}i} = \text{Interpolate}(S_{i},S_{j_{Fr}},\text{AllMatches}[i])$}
		\State{Compute $\tilde{S}_{i} = \text{Reparametrize}(S_{i},S_{j_{Fr}},\tilde{f}_{j_{Fr}i})$}
	\EndIf
\EndFor

\State{Output: $\tilde{S}_{1},...,\tilde{S}_{n}$}

\end{algorithmic}
\end{algorithm}

The formal pipeline we propose is titled EOP-Proj and is given by Algorithm~\ref{alg:MappingsFromAlignment}. Though the pipeline does not explicitly depend on aligned surfaces (and hence, projections after alignment), we restrict our attention to this setting. The algorithm works by first extracting two sets of points of interest from each shape: local maxima of some curvature (with radius of local maxima user-defined), and Gaussian Process landmarks. We pick as a template shape the Frechet mean with respect to the input distance matrix so as to have a canonical, data-driven template for which to base all registration off. We then run the EOP-PM on both sets of landmarks (the curvature extrema and the Gaussian process landmarks) to obtain two sets of matched landmarks, though using different mismatch thresholds for each set of landmarks. We are then left with two sets of sparse correspondences. 

Since both sets of landmarks are generated independently of one another, there is a significant possibility that both collections of registered pairs overlap. In order to deal with this overlap, we employ the following procedure. Every registered pair of curvature extrema is admitted into our collection of registered landmarks. Additional registered pairs are added iteratively based on maximizing the following functional, for $A_{i} = \{(a_{1}^{k},a_{2}^{k})\} \subseteq S_{i} \times S_{j_{Fr}}$: 

\begin{equation}\label{eqn:Energy}
E(p,A_{i}) =  \min_{(a_{1},a_{2}) \in A_{i}} d_{S_{i}}(p,a_{1}) +  \min_{(a'_{1},a'_{2}) \in A_{i}} d_{S_{j_{Fr}}}(p,a'_{2}) 
\end{equation}

\noindent where $d_{S_{i}}$ and $d_{j_{Fr}}$ are the geodesic distances on $S_{i}$ and $S_{j_{Fr}}$ respectively. In other words, we greedily add more registered pairs based on how far each pair is away from previously admitted pairs. Our choice of starting with registered pairs stemming from curvature extrema is motivated by empirical observations for small collections of nonisometric surfaces, where homologous curvature maxima may be located farther away from one another, which may not be properly accounted for in Gaussian Process landmark registration. Thus, we can think of the addition of Gaussian Process landmark pairs as filling in details unaccounted for by the curvature extrema.

After obtaining the sparse registration, we proceed by interpolating these pairs to obtain full registrations of the surfaces with the Frechet mean. Once these are obtained, we then reparametrize each of the surfaces so that the mesh structure is the same as that of the Frechet mean. This is possible if the registrations computed are homeomorphisms; though homeomorphisms are not guaranteed in general, some methods for computing full registrations are guaranteed to be locally homeomorphic (e.g. \cite{aigerman2016hyperbolic}). We do not, however, propose a specific procedure for this that we always follow, and leave that for the user to determine. At this point, all surfaces are consistently registered as all registrations follow naturally from the mesh structure of the Frechet mean.

\subsection{Experimental Results}

\subsubsection{Tooth Crowns}
\begin{figure*}
\begin{center}
\centering
\includegraphics[width=\textwidth]{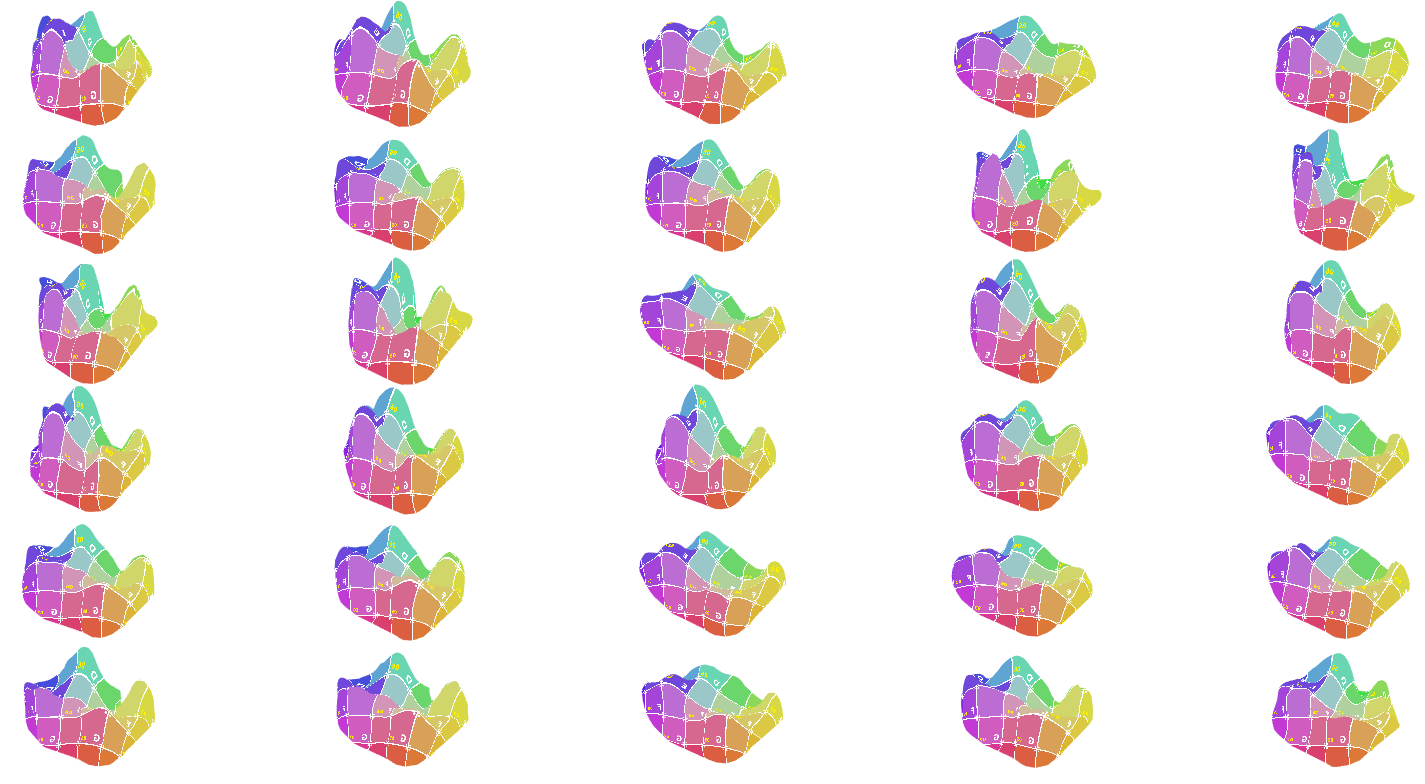}
	\captionof{figure}{Thirty surfaces from \cite{boyer2011algorithms} that are consistently registered using Algorithm~\ref{alg:MappingsFromAlignment}. The registrations are indicated by the textures.}\label{fig:TexturePNAS}
\end{center}
\end{figure*}

Continuing our previous experiments, we apply Algorithm~\ref{alg:MappingsFromAlignment} to the entire collection of tooth crowns from \cite{boyer2011algorithms} to generate a consistently registered collection of surfaces. All surfaces are first aligned via Auto3DGM \cite{boyer2015new}. Here, we use as curvature the conformal factor corresponding to the mid-edge uniformization \cite{lipman2009mobius} along with 400 Gaussian process landmarks, and restrict the total collection of matches used in registration to at most 14. As the meshes have reasonably uniform triangulation, we define our threshold radii in terms of the discrete distance based on the triangulation (i.e. the number of edges between one vertex and another): we use a radius of 8 for the conformal factor maxima, and 1 for the Gaussian Process landmarks. All parameters for the soft correspondence are the same default parameters used in Section~\ref{sec:registration}. The surfaces were all embedded into the plane using the same landmark constraint methodology employed in \cite{gao2019gaussian}. We show texture maps on thirty of the consistently  registered surfaces in Figure~\ref{fig:TexturePNAS}.

\begin{table*}
\begin{center}
\begin{tabular}{c|c|c|c|c|c|c}
\toprule
Method & Expert  & cP  &Auto3DGM  & PM-SDP  &EOP-Proj-FPS-200 (OURS)&EOP-Proj-GP-200 (OURS)  \\
\midrule
Genus &91.9 &90.9 &89.9 & 91.9 &90.9 & {\bf{95.0}} \\
Family & 94.3& 92.5& 90.5 & 94.3& 92.4 & {\bf{96.2}} \\
Order & 95.7& 94.8& 92.2 & {\bf{98.2}}& 96.6 & 97.2\\
\bottomrule
\end{tabular}
\end{center}
\caption{Percent accuracy of leave one out classification of 116 teeth from \cite{boyer2011algorithms} for a variety of algorithms. Registration via EOP-Proj achieves state of the art results in classification by both genus and family, and competitive results for classification by order. Note that these results are better than those achieved by human experts.}\label{tab:PNASCompare}
\end{table*}

As a test of the quality of our registrations, we repeat the leave one out classification benchmark introduced in \cite{boyer2011algorithms} and report the results in Table~\ref{tab:PNASCompare}. Each method in the table naturally yields a distance matrix, making leave one out classification a matter of finding the smallest nonzero element of each row. Our methods are given by EOP-Proj-FPS-200 and EOP-Proj-GP-200. To generate our distance matrix for EOP-Proj-FPS-200, we first compute the conformal factor extrema of the Frechet mean surface used in Algorithm~\ref{alg:MappingsFromAlignment}, and then feed those points into either the usual farthest point sampling algorithm until we have a collection of 200 points. We then take, from each surface in our colletion, the corresponding 150 points, and then compute the usual Procrustes distance (after normalization and alignment to the point cloud of the Frechet mean). The sample points used in EOP-Proj-GP-200 are the corresponding points after selecting 200 Gaussian Process landmarks on the Frechet mean surface.

We compare our method to a number previously studied on this example: expert landmarks \cite{boyer2011algorithms}, continuous Procrustes registration \cite{lipman2013continuous}, Auto3DGM \cite{boyer2015new}, and PM-SDP \cite{maron2016point}, a more recent method based on semi-definite programming that, to the best of our knowledge, achieved the best prior results on this dataset. We see in the table that EOP-Proj-GP-200 achieves state of the art results for both classification by genus and family, and is better than expert landmarks for classification by order. This heavily suggests that our method is able to better capture the fine level details of shape space than other automatic and expert-based methods. That the Gaussian Process method outperforms the farthest point sampling method gives further credence to the idea that the way in which the points are subsampled over the surface matters

\subsubsection{Talus bones}

\begin{figure*}
\begin{center}
\centering
\includegraphics[width=\textwidth]{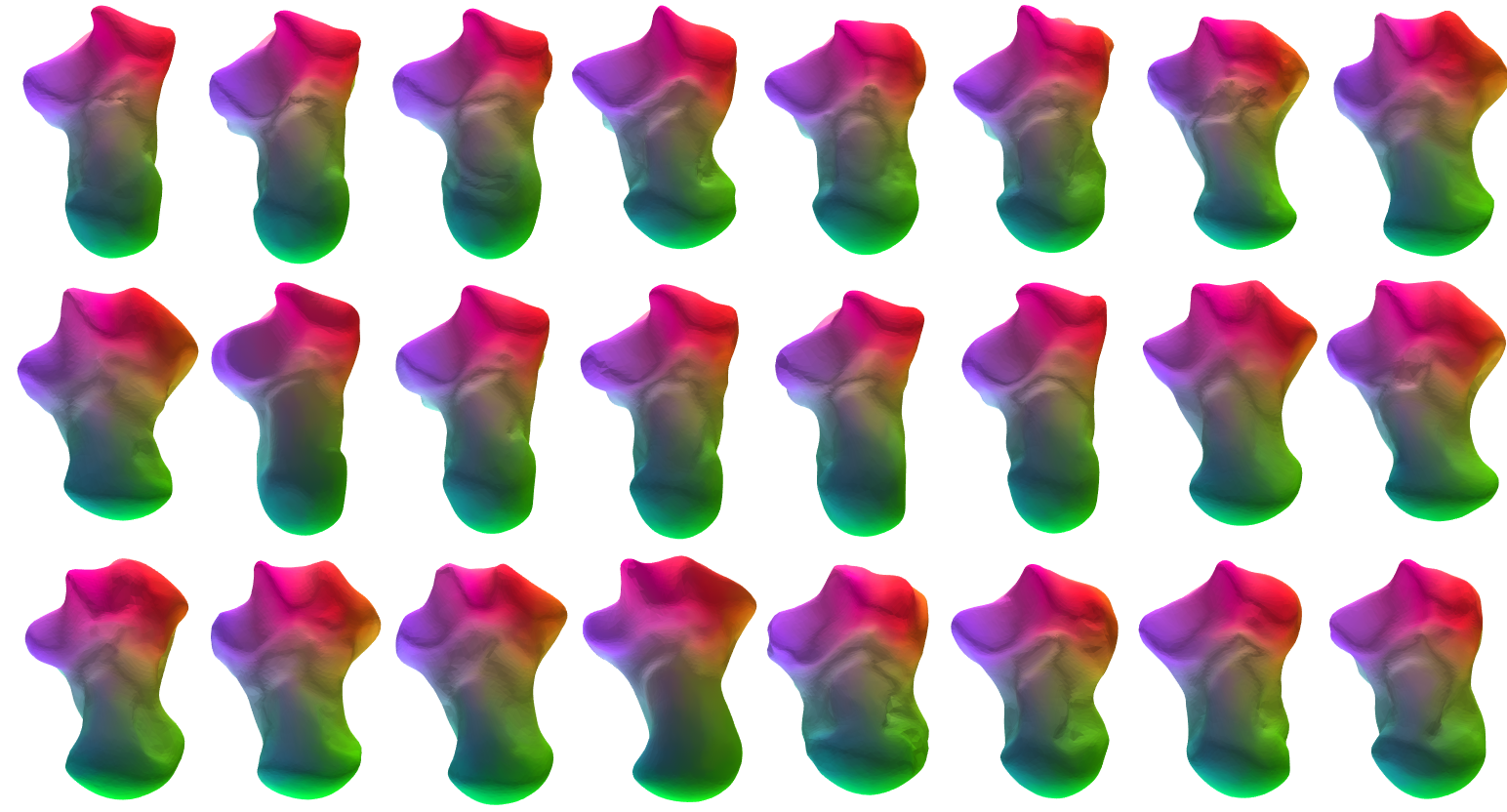}
\\
\includegraphics[width=\textwidth]{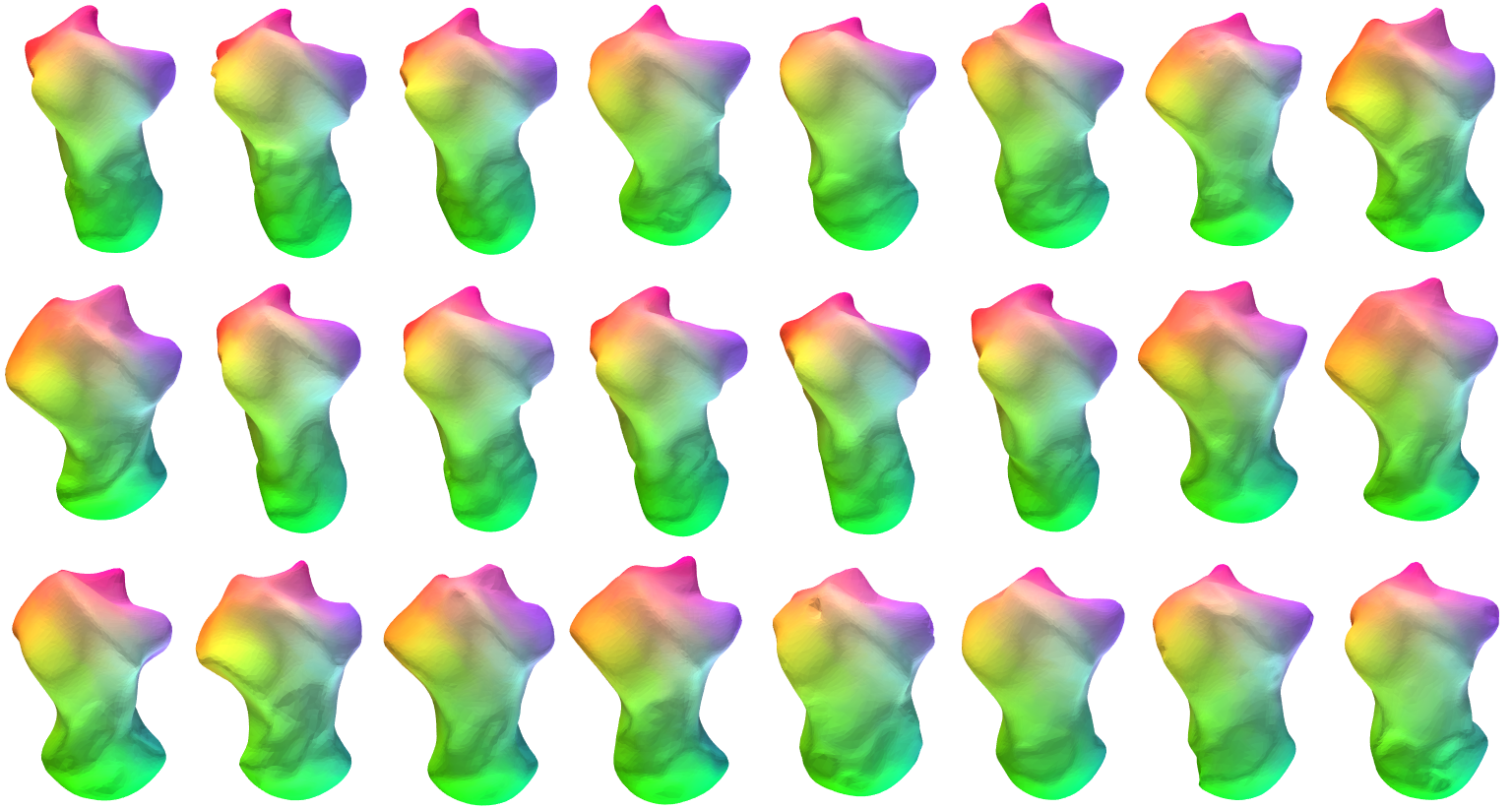}
	\caption{Two views of twenty four consistently registered talus specimens using Algorithm~\ref{alg:MappingsFromAlignment}. Registrations are determined by color.}\label{fig:TalusColor}
\end{center}
\end{figure*}

\begin{figure}[h]
\includegraphics[width=0.5\textwidth]{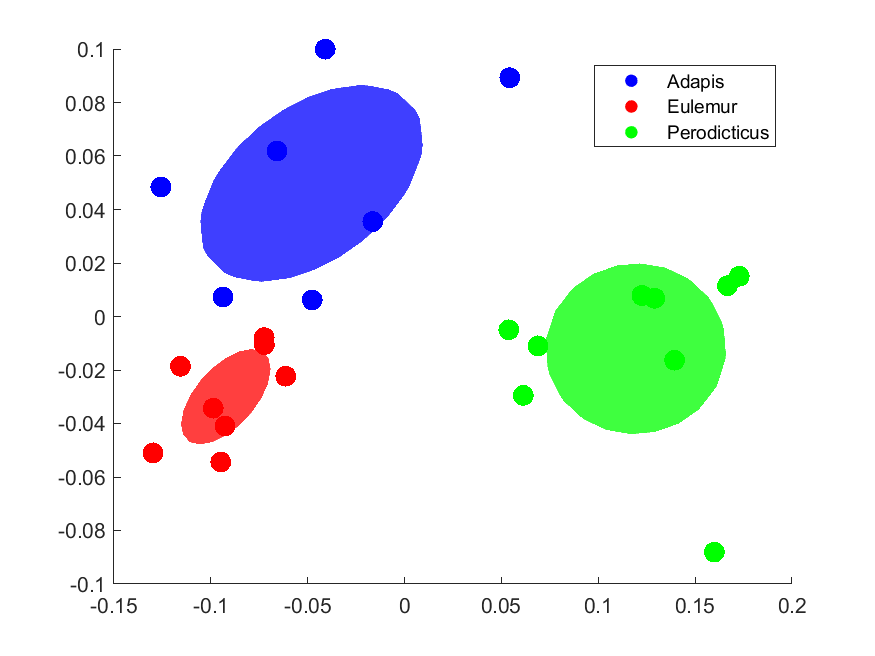}

	\caption{A two dimensional PCA plot of the surfaces in Figure~\ref{fig:TalusColor} with respect to the Procrustes distance.} \label{fig:TalusPCA}

\end{figure}

As previously mentioned, there is no prior restriction of our methods on surface topology. To illustrate this, we consider a small sample of talus bones from three genera (Adapis, Eulemur, and Perodicticus) that were obtained from Morphosource \cite{boyer2016morphosource}. These surfaces were registered using Gaussian curvature maxima and 200 Gaussian process landmarks. Registered landmarks were interpolated by the hyperbolic orbifold method proposed in \cite{aigerman2016hyperbolic}. Registrations are shown in Figure~\ref{fig:TalusColor}, showing homologous curvature extrema in correspondence. We can also see in Figure~\ref{fig:TalusPCA} the result of a two dimensional PCA embedding of these surfaces with respect to the Procrustes distances. The distances generated by EOP-Proj are capable of clustering each genus, and showing that both Adapis and Eulemur are more closely related to each other than either are to Perodicticus despite using a relatively low dimensional embedding.

\subsection{Limitations}
The main limitation of EOP-Proj and variants that use some other class of registrations, as well as the other algorithms proposed in the paper, is the implicit requirements on the datasets at hand. Namely, the implicit assumption that the collection of shapes is sufficiently densely sampled is not a ubiquitous property of many datasets of interest within the realm of registration (though many in addition to the ones considered above, such as DynaFAUST \cite{bogo2017dynamic} and faces \cite{ma2015non, gilani2017dense}, at least somewhat satisfy this property). Without this assumption, the differential geometric ideas on which we built our algorithms breaks down. We do not, however, see this as a total negative, especially in light of the promising state of the art results achieved by the above methods. The methods proposed are a natural way to deal with registration of nonisometric shapes, in which pure optimization can fail. We thus see our methodology as addressing an important use case that had previously gone unaddressed within the realm of registration.

Another limitation of EOP-Proj is the particular way in which the registrations are made consistent. Our use of the Frechet mean, though a reasonably heuristic, is by no means guaranteed to be optimal. An easy thought experiment suggests potential problems can arise if the Frechet mean itself has radically different geometry than other shapes in the collection. Alternative methods to guarantee consistency are thus subjects for future work.
\section{Conclusion}\label{sec:conclusion}
In this paper, we investigated the problem of consistent registrations of nonisometric surfaces. Motivated by biology, manifold learning and the previous role of parallel transport within registration, we proposed soft registration via EOP flows (EOP-SR), which have appealing properties both theoretically and computationally. We showed that employing EOP-SR can drastically improve accuracy over registration based purely on optimization methods while being robust by registering a point on one shape to a distribution of points on another shape. We adapted our methodology to build a partial matching method, EOP-PM, that can fix even subtle issues in registration quality. We then showed how EOP-PM can be used to accurately and consistently register large collections of shapes, achieving state of the art classification results in the process.

There are many directions of interest in which to adapt our work. In the course of our paper, interesting theoretical questions concerning parallel transport arose; progress on these questions could yield insight for guarantees about registration quality. From a computational perspective, our focus on surfaces from evolutionary biology was purely of interest and not a restriction. In particular, other datasets, such as face scans and certain collections of images, are natural candidate in which to adapt the algorithms of our paper, which we leave for future work. Finally, as mentioned in the Section~\ref{sec:consistent}, the ACC condition proposed in the definition of EOP-PM has potential merit for further investigation given the recent increase that cycle consistency has received within the general machine learning literature.

\subsection{Declarations}
The author was partially supported by AFOSR through an NDSEG Research Fellowship. The author declares no conflicts of interest. Availability information of data and code relevant to the above discussions is detailed in the Appendix. The author would like to thank Ingrid Daubechies, Tingran Gao, Shahar Kovalsky, Chen-Yun Lin, Sayan Mukherjee, Shan Shan, and Barak Sober for many helpful discussions in compiling this work.

\bibliographystyle{plainnat}
\bibliography{bibliography}

\newpage

\appendix
\onecolumn
\section{Appendix}\label{sec:appendix}
\subsection{Further Review of Differential Geometry}
We briefly elaborate on details relevant to discussions in Sections~\ref{sec:prelim} and~\ref{sec:fiber}. For brevity, we refer to readers largely unfamiliar with smooth manifold theory and Riemannian geometry to standard texts previously mentioned, namely \cite{chavel2006riemannian, michor2008topics, lee2013smooth}, and review only details most relevant for our work.

\subsubsection{Riemannian Geometry}
We let $(M^{m},g)$ be an $m$-dimensional Riemannian manifold with metric $g$ and tangent bundle $TM.$ We also denote by $\mathcal{V}(M)$ the space of smooth vector fields on $M$. Recall that vector fields naturally as differential operators on smooth functions; e.g. if a vector field $X = \sum_{i=1}^{m} X_{i} \frac{\partial}{\partial x_{i}}$ is an expansion of $X$ in local coordinates, then for a smooth function $f$ on $M,$ we have $Xf = \sum_{i=1}^{m} X_{i} \frac{\partial f}{\partial x_{i}}.$

\begin{definition}
A {\bf{connection}} on a Riemannian manifold $(M,g)$ is a map $\nabla: \mathcal{V}(M) \times \mathcal{V}(M) \to \mathcal{V}(M),$ written $(X,Y) \to \nabla_{X} Y,$ that satisfies the following three properties:

\begin{enumerate}
\item For $f, g$ smooth functions on $M$ and $X, Y, Z$ smooth vector fields, we have

$$\nabla_{fX + gY} Z = f \nabla_{X} Z + g \nabla_{Y} Z$$

\item For $a, b \in \mathbb{R}$ and $X,Y,Z$ smooth vector fields on $M,$ we have

$$\nabla_{X} (aY+bZ) = a \nabla_{X}Y + b \nabla_{X}Z$$

\item For $f$ a smooth function on $M$ and $X,Y$ smooth vector fields, we have a Leibniz rule

$$\nabla{X}(fY) = f \nabla_{X}Y + (Xf)Y$$

\noindent where $X$ acts a partial derivative on $f$
\end{enumerate}
\end{definition}

\noindent Connections are the appropriate analogue of a directional derivative for vector fields. Specifically, for a vector field $X$ along a curve $\gamma:[a,b] \to M,$ $X'(t) = \frac{d}{dt}X := \nabla_{\dot{\gamma}(t)},$ sometimes written as $\nabla_{\frac{d}{dt}} X.$ It is well-known in Riemannian geometry that every manifold has a unique connection called the {\bf{Levi-Civita connection}} that satisfies the above as well as two additional properties: it is symmetric, i.e.

$$\nabla_{X}Y -\nabla_{Y}X = [X,Y]$$

\noindent where $[\cdot, \cdot]$ is the Lie bracket on $TM,$ the tangent bundle of $M,$ and it is compatible with the metric, i.e.

$$\nabla_{X}g(Y,Z) = g \left( \nabla_{X}Y,Z \right) + g \left(X, \nabla_{Y} Z \right)$$

The Levi-Civita connection satisfies the following property, which we need in proving Theorem~\ref{thm:EOPEstimate}. Note, however, that we only use the compatibility property.

\begin{lemma} \label{DerivLemma}
Let $X(t)$ be a smooth vector field defined along a curve $\gamma: [a,b] \to M$ nowhere equal to zero. Then
we have

$$\frac{d}{dt}||X|| := \frac{d}{dt}(g(X,X))^{\frac{1}{2}} \leq \left( g \left(X,X \right) \right)^{\frac{1}{2}} = ||\nabla_{\frac{d}{dt}}X||$$

\noindent where all derivatives are taken with respect to the Levi-Civita connection.
\end{lemma}
\begin{proof}
Since the Levi-Civita connection is symmetric, we have

\begin{align*}
\frac{d}{dt}(g(X,X))^{\frac{1}{2}} &= \frac{1}{2}g(X,X)^{- \frac{1}{2}} \frac{d}{dt} g(X,X) \\
&=\frac{1}{2}g(X,X)^{- \frac{1}{2}} 2 g(\nabla_{\frac{d}{dt}}X,X) \\
&=g(X,X)^{- \frac{1}{2}}g(\nabla_{\frac{d}{dt}}X,X)
\end{align*}
\noindent by symmetry of the Riemannian metric and compatibility of the Levi-Civita connection. A direct application of the Cauchy-Schwarz inequality on $g(X',X)$ completes the proof.
\end{proof}

We assume for the remainder that we are only working with the Levi-Civita connection.

\subsubsection{Parallel Transport}

The heart of our methodology relies on parallel transport. The usual definition pertaining to vector fields is easy to state given for a Levi-Civita connection; given a smooth curve $\gamma:[a,b] \to M,$ there is a natural vector field $\cdot{\gamma}$ defined on the image of $\gamma$ given by the usual time derivative. Then, the parallel transport of a vector $v_{\gamma(a)} \in T_{\gamma(a)}M$ along $\gamma$ is the unique vector field $V$ defined on $\gamma(t)$ such that $V(a) = v$ and

$$\nabla_{\cdot{\gamma}(t)}V(t) = 0.$$

For fiber bundles, the standard definition of parallel transportis similar but quite technical, requiring a number of definitions that we would not directly reference again. In light of this, we adopt an equivalent definition, in which the notion of connection is itself abstracted away from the usual language. Despite this abstraction, the concepts are much more intuitive and clearly apply to our setting. Let $H(F)$ denote the space of homeomorphisms (or diffeomorphisms, as appropriate) of $F.$

\begin{definition}\label{def:connectionFiber}
Let $(E,B,F,\pi)$ be a fiber bundle with base manifold $B$ and fiber $F.$ A {\bf{connection}} is a function that maps every curve $\gamma:[0,1] \to B,$ to a family of homeomorphisms of $F$ (or diffeomorphisms, if the fiber bundle is smooth) $P_{\gamma}: [0,1] \times [0,1] \to H(F)$ such that $P_{\gamma}$ is continuous on $[0,1] \times [0,1],$ $P_{\gamma}(t,t) = Id_{F},$ and satisfies the following transitivity property for all $a,b,c \in [0,1]$:

$$ P_{\gamma}(c,b) \circ P_{\gamma}(b,a) = P_{\gamma}(c,a) $$
\end{definition}

\begin{definition}\label{def:parallelFiber}
Let $(E,B,F,\pi)$ be a fiber bundle with base manifold $B$ and fiber $F.$ A {\bf{connection}} is a function that maps every curve $\gamma:[0,1] \to B,$ to a family of homeomorphisms of $F$ (or diffeomorphisms, if the fiber bundle is smooth) $P_{\gamma}: [0,1] \times [0,1] \to H(F)$ such that $P_{\gamma}$ is continuous on $[0,1] \times [0,1],$ $P_{\gamma}(t,t) = Id_{F},$ and satisfies the following transitivity property for all $a,b,c \in [0,1]$:

$$ P_{\gamma}(c,b) \circ P_{\gamma}(b,a) = P_{\gamma}(c,a) $$

\noindent Given such a connection, the {\bf{parallel transport}} of $\gamma$ is given by $P_{\gamma}(1,0).$
\end{definition}

\noindent This property can be derived from the usual definition of parallel transport on fiber bundles as present in \cite{michor2008topics}. Of particular interest, however, is that our definition works for more general cases when smooth structures don't exist, and would readily apply to spaces of meshes.

\subsubsection{Fermi normal coordinates}

The last piece of exposition we require is Fermi normal coordinates, which were mentioned briefly in Theorem~\ref{thm:EOPFunction}. Fermi normal coordinates are the natural analogue to Cartesian coordinates in Euclidean space as exponential coordinates are the manifold analogue of polar coordinates in Euclidean space.

\begin{definition}
Let $P^{q} \subset M^{m}$ be a regular submanifold, let $x_{1},...,x_{q}$ be a coordinate chart on $P,$ and let $E_{q+1},...,E_{m}$ be a section of the tangent bundle along $P$ such that the $E_{j}|_{p}$ are an orthonormal completion of $T_{p}P$ in $T_{p}M.$ Then the functions $y_{i}, 1 \le i \le m$ defined by

\begin{align*}
y_{i} \left( \exp_{p} \left( \sum_{j=q+1}^{m} a_{j}E_{j}|_{p} \right) \right) = x_{i}(p), 1 \le i \le q \\
y_{i} \left( \exp_{p} \left( \sum_{j=q+1}^{m} a_{j}E_{j}|_{p} \right) \right) = a_{i}, q+1 \le i \le m
\end{align*}

form a {\bf{Fermi normal coordinate chart}} with respect to $P.$
\end{definition} \label{def:FermiNormalCoords}

\noindent Though the definition we state is for the general setting, we are only interested in the case where $q = 1$ in the above definition. Thus, in our work, we let $P$ be a geodesic between two points in a manifold. Theorem~\ref{thm:EOPFunction} that there is a family of curves that can be written as functions along the geodesic axis of such a Fermi coordinate chart.

\subsection{Proofs}
\subsubsection{Proof of Theorem~\ref{thm:EOPEstimate}}

By the fundamental theorem of calculus, we have

\begin{align*}
\Vert X_{1}(1)-X_{0}(1) \Vert &= \Vert  \int_{0}^{1}  \nabla_{\frac{d}{ds}}X_{s}(1) ds \Vert \\
&\leq   \int_{0}^{1} \Vert  \nabla_{\frac{d}{ds}}(X_{s}(1)) \Vert ds   \\
&= \int_{0}^{1} \left( \Vert  \nabla_{\frac{d}{ds}}(X_{s}(1)) \Vert- \Vert(\nabla_{\frac{d}{ds}}X_{s}(0)) \Vert \right)ds
\end{align*}

\noindent where the last equality follows from the assumption that $X_{s}(0)$ is constant. Another application of the fundamental theorem gives

\begin{align*}
\Vert X_{1}(1)-X_{0}(1) \Vert &\leq \int_{0}^{1} \left( \Vert  \nabla_{\frac{d}{ds}}(X_{s}(1)) \Vert- \Vert(\nabla_{\frac{d}{ds}}X_{s}(0)) \Vert \right)ds \\
&= \int_{0}^{1} \left( \int_{0}^{1}  \frac{d}{dt} \Vert \nabla_{\frac{d}{ds}}(X_{s}(t)) \Vert dt \right)ds \\
&\leq \int_{0}^{1} \left( \int_{0}^{1} \Vert \nabla_{\frac{d}{dt}} \nabla_{\frac{d}{ds}}(X_{s}(t)) \Vert dt \right)ds \\
&= \int_{0}^{1} \left( \int_{0}^{1} \Vert (\nabla_{\frac{d}{dt}} \nabla_{\frac{d}{ds}} - \nabla_{\frac{d}{ds}} \nabla_{\frac{d}{dt}}) (X_{s}(t)) \Vert dt \right)ds \\
\end{align*}

\noindent where the second inequality follows from Lemma \ref{DerivLemma} and the last equality follows from the definition of parallel transport. Since $\frac{d}{ds}$ and $\frac{d}{dt}$ are coordinate vector fields, we have from Definition \ref{Riemann}

$$\Vert X_{1}(1)-X_{0}(1) \Vert \leq \int_{0}^{1} \left( \int_{0}^{1} \Vert (R(H_{t},H_{s})(X_{s}(t)) \Vert dt \right)ds$$

\noindent where $H_{t}$ and $H_{s}$ are the corresponding partial derivative vector fields of the homotopy $H.$ The Riemann curvature term in the above integral is bounded above by $\frac{4}{3} K_{max} \Vert H_{s} \wedge H_{t} \Vert,$ where $\Vert H_{s} \wedge H_{t} \Vert$ is the induced area form of the homotopy; this is a direct consequence of Lemma 3.7 in \cite{bourguignon1978curvature}. The desired result immediately follows.

\subsubsection{Proof of Theorem \ref{thm:EOPFunction}}
Assume that the inner products of $s'(t)$ with $v_{s(t),x}$ and $v_{s(t),y}$ are strictly bounded above and below respectively by some uniform numbers $M_{1}$ and $M_{2}$ whenever the inner products are well defined, i.e. when $s(t) \neq x, y.$ We now define two sets: let $U_{x,\varepsilon}$ be the subset of $M$ such that every point $u \in U_{x,\varepsilon}$ satisfies $g \left( v_{x,u}, \frac{\partial}{\partial \gamma} \right) \geq 1 - \varepsilon,$ and analogously  $U_{y,\varepsilon}.$ If $M_{1}$ and $M_{2}$ are at least $\varepsilon,$ then any such curve $s(t)$ contained in the intersection of the Fermi coordinate neighborhood of $\gamma$ with the union of $U_{x,\varepsilon}$ and $U_{y,\varepsilon}$ necessarily has to satisfy $g \left( s'(t), \frac{\partial}{\partial \gamma} \right) > 0.$ An application of the implicit function theorem then completes the proof.

\subsection{Experiment Details}
\subsubsection{Meshes}
All meshes used throughout this paper are publicly available on the Morphosource Repository \cite{MorphosourceWeb}. All specimens were cleaned to be either homeomorphic to a disk or sphere as appropriate and downsampled to approximately 5000 vertices.
\subsubsection{Code}
We made use of the following external repositories for methods referenced in the text.

\begin{itemize}
\item{Continuous Procrustes Maps: {\emph{https://github.com/trgao10/cPdist}}}  \item{GP-BD Registration: \emph{https://github.com/shaharkov/GPLmkBDMatch}} \item{Auto3DGM: \emph{https://github.com/trgao10/PuenteAlignment}}
\end{itemize}
\noindent Default parameters were used. 

\noindent The code used to run the experiments of the paper was derived from the SAMS repository found at {\emph{https://github.com/RRavier/SAMS}}, in which the registration method is an implementation of EOP-Proj written by the author of this work. A separate repository, which {\emph{https://github.com/RRavier/EOP-Experiments}}, will contain a more concise version focused on replicating the experiments found in the paper.
\end{document}